%% file: main.tex
\documentclass[10pt]{article}
\usepackage{pgfplots}
\pgfplotsset{compat=1.14}
\usepackage{tikz-3dplot}
\usepackage{xspace}
\usepackage{tikz}
\usepackage{multirow}
\usepackage{marvosym}
\usepackage{enumitem}
\usepackage{fullpage}
\usepackage{amsthm}
\usepackage{amssymb}
\usepackage{mathtools}
\usepackage{amsfonts}
\usepackage{booktabs}
\usepackage{subcaption}
\usepackage{relsize}
\usepackage[T1]{fontenc}

\usepackage{hyperref}
\pgfplotsset{compat=1.16}
\usepackage{tikz-3dplot}
\usepackage{xspace}
\usepackage{tikz}
\pgfdeclarelayer{background}
\pgfsetlayers{background,main}

\usepackage{xcolor,colortbl}

\input{myMacros.tex}

\date{}
\title{Tractable Conjunctive Queries over Static and Dynamic Relations}

\author{
\begin{tabular}{ccc}
\hspace{0.1cm} Ahmet Kara \hspace{0.5cm}& \hspace{0.5cm} Zheng Luo \hspace{0.5cm} & \hspace{0.5cm} Milos Nikolic \hspace{0.5cm} \\
\hspace{0.1cm} ahmet.kara@oth-regensburg.de \hspace{0.5cm} & \hspace{0.5cm} luo@cs.ucla.edu \hspace{0.5cm} & \hspace{0.5cm} milos.nikolic@ed.ac.uk \hspace{0.5cm} 
\end{tabular}\\ \\
\begin{tabular}{cc}
 Dan Olteanu \hspace{0.5cm} & \hspace{0.5cm} Haozhe Zhang \hspace{0.5cm} \\
dan.olteanu@uzh.ch \hspace{0.5cm} & \hspace{0.5cm} haozhe.zhang@uzh.ch \hspace{0.5cm} 
\end{tabular}
}

\begin{document}
\maketitle

\begin{abstract}
We investigate the evaluation of conjunctive queries over static and dynamic relations. While static relations are given as input and do not change, dynamic relations are subject to inserts and deletes.

We characterise syntactically three classes of queries 
that admit constant update time and constant enumeration delay. 
We call such queries {\em tractable}. Depending on the class, the preprocessing time is linear, polynomial, or exponential (under data complexity, so the query size is constant).

To decide whether a query is tractable, it does not suffice to analyse separately the sub-queries over the static relations and over the dynamic relations, respectively. Instead, we need to take the interaction between the static and the dynamic relations into account. Even when the sub-query over the dynamic relations is not tractable, the overall query can become tractable if the dynamic relations are sufficiently constrained by the static ones.      
\end{abstract}

\paragraph{Acknowledgements.}
The authors would like to acknowledge the UZH Global Strategy and Partner-
ships Funding Scheme and the EPSRC grant EP/T022124/1. Olteanu would like to thank the IVM
team at RelationalAI and in particular Niko Göbel for discussions that motivated this work. Luo’s
contribution was done while he was with the University of Zurich.

\input{intro}

\input{prelims}
\input{rewriting}
\input{query_classes}
\input{evaluation_lin_pol}
\input{evaluation_expo}
\input{lower_bounds}
\input{beyond_exp}

\bibliographystyle{abbrv}
\bibliography{bibliography}

\appendix
\input{app_intro}
\input{app_rewriting}
\input{app_query_classes}
\input{app_evaluation_lin_pol}
\input{app_evaluation_expo}

\input{app_lower_bounds}

\end{document}

%% file: myMacros.tex
\usepackage{amsmath}
\usepackage{hyperref}
\usepackage{pgfplots}
\pgfplotsset{compat=1.16}
\usepackage{tikz-3dplot}
\usepackage{xspace}
\usepackage{tikz}
\pgfdeclarelayer{background}
\pgfsetlayers{background,main}
\usetikzlibrary{decorations.pathreplacing,calligraphy,backgrounds,positioning,arrows.meta,shapes.multipart,matrix}

\usepackage{float}
\usepackage{multirow}
\usepackage{marvosym}
\usepackage{booktabs}

\usepackage{tcolorbox}
\usepackage{mdframed}
\mdfdefinestyle{MyFrame}{
outerlinewidth=2pt,
roundcorner=5pt,
leftmargin=100pt,
innertopmargin=\baselineskip,
innerbottommargin=1pt,
innerrightmargin=20pt,
innerleftmargin=100pt,
backgroundcolor=gray!30!white
}
\usepackage{minibox}
\usepackage{xcolor,colortbl}
\usepackage{bm}

\usepackage{xspace}
\usepackage{algorithmicx}
\usepackage{algpseudocode}

\usepackage[ruled,linesnumbered]{algorithm2e} 

\newcommand{\lin}{\calC_{\texttt{lin}}}
\newcommand{\pol}{\calC_{\texttt{poly}}}
\newcommand{\expo}{\calC_{\texttt{exp}}}

\definecolor{light-gray}{gray}{0.7.2}
\definecolor{goodgreen}{rgb}{0.1, 0.5, 0.1}
\definecolor{burntorange}{rgb}{0.8, 0.33, 0.0}

\newcommand{\combine}{\textsf{Combine}}
\newcommand{\rewrite}{\textsc{Rewrite}}

\newcommand{\vars}{\mathit{vars}}

\newcommand{\free}{\mathit{free}}

\newcommand{\atoms}{\mathit{atoms}}
\newcommand{\dynatoms}{\mathit{dynAtoms}}

\newcommand{\dyn}{\mathit{dyn}}
\newcommand{\stat}{\mathit{stat}}

\newcommand{\dep}{\textit{dep}}

\newcommand{\bigO}[1]{\mathcal{O}(#1)}

\newcommand{\WVO}{\mathsf{WVO}}

\newcommand{\fhtw}{\mathsf{fhtw}}
\newcommand{\fw}{\mathsf{w}}

\newcommand{\linenumber}{\makebox[2ex][r]{\rownumber\TAB}}

\newcommand{\tup}[1]{\mathbf{#1}}

\newcommand{\calN}{\mathcal{N}}

\newcommand{\calT}{\mathcal{T}}

\newcommand{\calB}{\mathcal{B}}
\newcommand{\calA}{\mathcal{A}}

\newcommand{\calC}{\mathcal{C}}

\newcommand{\calS}{\mathcal{S}}

\newcommand{\calX}{\mathcal{X}}

\newcommand{\calU}{\mathcal{U}}

\newcommand{\TAB}{\makebox[2.5ex][r]{}}%
\newcommand{\SPACE}{\makebox[0.75ex][r]{}}%
\newcommand{\LET}{\textbf{let}\xspace}%
\newcommand{\IF}{\textbf{if}\xspace}%
\newcommand{\FOREACH}{\textbf{foreach}\xspace}%
\newcommand{\RETURN}{\textbf{return}\xspace}%
\newcommand{\MATCH}{\textbf{switch}\xspace}%
\usepackage[colorinlistoftodos]{todonotes}

\newcommand{\nop}[1]{}

\newcounter{magicrownumbers}
\newcommand\rownumber{\footnotesize\stepcounter{magicrownumbers}\arabic{magicrownumbers}}

\newcommand{\change}[1]{{\color{blue} #1}}

\theoremstyle{plain}                  
\newtheorem{theorem}{Theorem}
\newtheorem{lemma}[theorem]{Lemma}
\newtheorem{proposition}[theorem]{Proposition}
\newtheorem{conjecture}[theorem]{Conjecture}
         
\newtheorem{definition}[theorem]{Definition}
\newtheorem{example}[theorem]{Example}

\newtheorem{claim}[theorem]{Claim}

%% file: intro.tex
\section{Introduction}
\label{sec:intro}
Incremental view maintenance, also known as fully dynamic query evaluation, is a fundamental task in data management~\cite{DBT:VLDBJ:2014, Idris:dynamic:SIGMOD:2017, DattaKMSZ18, Kara:TODS:2020, TaoY22, FIVM:VLDBJ:2023, DBLP:journals/pvldb/WangHDY23, BudiuCMRT23}.
A natural question is to understand which conjunctive queries are tractable, i.e., which queries admit constant time per single-tuple update (insert or delete) and also constant delay for the enumeration of the result tuples. 
The problem setting also allows for some one-off preprocessing phase to construct a data structure that supports the updates and the enumeration.
Prior work showed that the $q$-hierarchical queries are the conjunctive queries without repeating relation symbols that are tractable; such queries admit constant update time and enumeration delay, even if we only allow for linear time preprocessing~\cite{BerkholzKS17}. All other conjunctive queries without repeating relation symbols cannot admit both constant update time and constant enumeration delay, even if we allow arbitrary time for preprocessing. If we only allow inserts (and no deletes), then every free-connex $\alpha$-acyclic conjunctive query becomes tractable~\cite{DBLP:journals/pvldb/WangHDY23,DBLP:journals/corr/abs-2312-09331}.
The tractable queries with free access patterns, where the free variables are partitioned into input and output, naturally extend the  $q$-hierarchical queries, which are queries without input variables~\cite{ICDT23_access_pattern}. 
Beyond this well-behaved query class, there are known intractable queries such as the triangle queries~\cite{Kara:ICDT:19,Kara:TODS:2020} and hierarchical conjunctive queries with arbitrary free variables~\cite{KaraNOZ2020,DBLP:journals/lmcs/KaraNOZ23} for which algorithms are known to achieve worst-case, yet not constant, optimal update time and enumeration delay.
Restrictions to the data and the update sequence can turn intractable queries into tractable ones: for instance, when the update sequence has a bounded enclosureness values~\cite{DBLP:journals/pvldb/WangHDY23} or when the data satisfies functional dependencies~\cite{FIVM:VLDBJ:2023}, degree bounds~\cite{BerkholzKS18_bounded_degree}, or more general integrity constraints~\cite{BerkholzKS18}.

All aforementioned works consider the {\em all-dynamic setting}, where all relations are updatable. In this work, we extend the tractability frontier by considering a {\em mixed setting}, where the input database can have both dynamic relations, which are subject to change, and static relations, which do not change. 
The mixed setting appears naturally in practice, e.g., in the retail domain where some relations (e.g., Inventory, Sales) are updated very frequently while others (e.g., Stores, Demographics) are only updated sporadically if at all.
\nop{The mixed setting appears naturally in real-world applications at RelationalAI (personal communication). For instance, in the retailer domain, the Inventory and Sales relations are updated very frequently, whereas the Stores and Demographics relations are updated very infrequently and can be considered static for a large time period. }
We show that by differentiating between static and dynamic relations, we can design efficient query maintenance strategies tailored to the mixed setting.

\begin{figure}[t]
\begin{center}
\begin{tikzpicture}

\draw[rounded corners=4pt, line width=1.2pt, color = red] (0,0.5) rectangle (5,1.4);
\node at(1.2,0.9){$\lin$};
\node at(4.3,0.9){$\bigO{N}$};

\draw[dashed, rounded corners=4pt, color = blue, 
line width = 1.2] (-0.3,0.2) rectangle (5.3,2.2);
\node at(1.2,1.8){$\pol$};
\node at(4.3,1.8){$\bigO{N^{\fw}}$};

\draw[rounded corners=4pt, line width=1.2pt, color = blue, dashed] (-0.6,-0.1) rectangle (5.6,3);
\node at(1.2,2.6){$\expo$};
\node at(4.3,2.6){EXPTIME};

\draw[rounded corners=4pt] (-0.9,-0.4) rectangle (5.9,4);
\node at(1.5,3.5){Conjunctive Queries};
\node at(4.3,3.5){$-$};

\end{tikzpicture}
\caption{
Classes $\lin \subset \pol \subset \expo$ of tractable conjunctive queries over static and dynamic relations and the corresponding preprocessing time (data complexity). $N$ is the current database size and $\fw$ is the preprocessing width.
The solid (red) border around $\lin$ states that there is a dichotomy between the queries inside and outside the class, shown in this paper. 
The dashed (blue) border around a class states that no dichotomy is known for queries inside and outside the class.
}
\label{fig:result_overview}
\end{center}
\end{figure}

\subsection*{Main Contributions}

We characterise syntactically three classes of tractable conjunctive queries depending on their preprocessing time, cf.\@ Figure~\ref{fig:result_overview}: $\lin \subset \pol \subset \expo$. 
The queries are categorised based on their structure and on the structure of their static and dynamic sub-queries. These classes are defined in Section~\ref{sec:query_classes}.

The class $\lin$ characterises the tractable queries with linear time preprocessing: 

\begin{theorem}
\label{thm:dichotomy_linear}
    Let a conjunctive query $Q$ and a database of size $N$. 
\begin{itemize}
\item If $Q$ is in $\lin$, then it can be maintained with $\bigO{N}$ preprocessing time,
$\bigO{1}$ update time, and $\bigO{1}$ enumeration delay. 

\item If $Q$ is not in $\lin$ and does not have repeating relation symbols, then 
it cannot be maintained with  $\bigO{N}$ preprocessing time, $\bigO{1}$
update time, and $\bigO{1}$ enumeration delay, 
unless the Online Matrix-Vector Multiplication or the Boolean Matrix-Matrix Multiplication conjecture fail.  
\end{itemize}
\end{theorem}

The upper bound in Theorem~\ref{thm:dichotomy_linear} relies on a rewriting of the query $Q$ using strata of views, where the views are defined using projections or joins of input relations and possibly other views (Section~\ref{sec:rewriting}). We call such rewritings {\em safe} if the views can be maintained in constant time for single-tuple updates and support constant-delay enumeration of the query result. We show that every $\lin$ query has a safe rewriting and the views can be computed in linear time (Section~\ref{sec:evaluation_lin_pol}).
The lower bound in Theorem~\ref{thm:dichotomy_linear} relies on two widely used conjectures: the Online Matrix-Vector Multiplication~\cite{Henzinger:OMv:2015} and the Boolean Matrix-Matrix Multiplication~\cite{BaganDG07}. The proof of the lower bound is outlined in Section~\ref{sec:lower_bounds}.

\begin{figure}[t]
    \centering
\begin{tikzpicture}
    \node at(0, 0.45) {$\underline{A}$};
    \node at(0,-0.45) {$\underline{B}$};
    \node at(1, 0.45) {$D$};
    \node at(1,-0.45) {$\underline{C}$};
    
    \draw[red] (0.5,0.45) ellipse (1cm and 0.3cm);
    \draw[red] (0,0) ellipse (0.3cm and 1cm);
    \draw[blue, dotted, line width=1.2] (0.5,-0.45) ellipse (1cm and 0.3cm);

\begin{scope}[xshift=2.8cm]
    \node at(0, 0.45) {$\underline{A}$};
    \node at(0,-0.45) {$B$};
    \node at(1, 0.45) {$\underline{D}$};
    \node at(1,-0.45) {$\underline{C}$};
    
    \draw[red] (0.5,0.45) ellipse (1cm and 0.5cm);
        \draw[red] (1,0.45) ellipse (0.3cm and 0.3cm);
    \draw[blue, dotted, line width=1.2] (0,0) ellipse (0.3cm and 1cm);
    \draw[blue, dotted, line width=1.2] (0.5,-0.45) ellipse (1cm and 0.3cm);
\end{scope}    

\begin{scope}[xshift=5.6cm]
\draw[blue, dotted, line width= 1.2] (0,0) ellipse (0.5cm and 1cm);
    \draw[red] (0,0.45) ellipse (0.3cm and 0.3cm);
    \node at(0, 0.45) {$\underline{A}$};

    \draw[red] (0,-0.45) ellipse (0.3cm and 0.3cm);
    \node at(0,-0.45) {$\underline{B}$};
\end{scope}    

\begin{scope}[xshift=7.3cm]
\draw[red,rotate=-35] (0.2,0.1) ellipse (0.33cm and 1cm);
    \node at(0.5, 0.35) {$\underline{A}$};
     \node at(0,-0.45) {$\underline{B}$};
    \node at(1.1,-0.45) {$\underline{C}$};
    \draw[blue, dotted, line width=1.2pt] (0.45,-0.45) ellipse (1cm and 0.33cm);
      \draw[red,rotate=35] (0.6,-0.4) ellipse (0.33cm and 1cm);
 
    \nop{
    \node at(0, 0.45) {$A$};
    \node at(0,-0.45) {$B$};
    \node at(0.8, 0.45) {$D$};
    \node at(1,-0.45) {$\underline{C}$};
    
    \draw[red] (0.5,0.45) ellipse (1cm and 0.3cm);
    \draw[red] (0,0) ellipse (0.3cm and 1cm);
    \draw[blue, dotted, line width=1.2] (0.5,-0.45) ellipse (1cm and 0.3cm);    
}
\end{scope}

\begin{scope}[xshift=10.2cm]
\draw[red,rotate=-35] (0.2,0.1) ellipse (0.33cm and 1cm);
    \node at(0.5, 0.35) {$A$};

    \draw[blue, dotted, line width=1.2pt] (0.45,-0.45) ellipse (1cm and 0.33cm);
    \node at(0,-0.45) {$\underline{B}$};
    \node at(1.1,-0.45) {$\underline{C}$};
      \draw[red,rotate=35] (0.6,-0.4) ellipse (0.33cm and 1cm);
\end{scope}

\begin{scope}[xshift=13cm]
\draw[red] (0,0) ellipse (0.5cm and 1cm);
    \draw[red] (0,0.45) ellipse (0.3cm and 0.3cm);
\draw[dotted, blue, line width = 1.2] (1.1,-0.45) ellipse (0.3cm and 0.3cm);    
    \node at(0, 0.45) {$\underline{A}$};
    \draw[red] (0.6,-0.45) ellipse (1cm and 0.45cm);
    \node at(0,-0.45) {$\underline{B}$};
    \node at(1.1,-0.45) {$C$};
\end{scope}

\node at(0.5,1.5){$Q_1\in \lin$};
\node at(3.3,1.5){$Q_2\in \pol$};
\node at(5.7,1.5){$Q_3\in \expo$};
\node at(7.8,1.5){$Q_4\not\in \expo$};
\node at(10.8,1.5){$Q_5\notin \expo$};
\node at(13.5,1.5){$Q_6\notin \expo$};


\begin{scope}[yshift=0.3cm]
\node at(3,-1.8){$Q_1(A,B,C) = R^d(A,D), S^d(A,B), T^s(B,C)$};
\node at(11,-1.8){$Q_2(A,C,D) = R^d(A,D), S^s(A,B), T^s(B,C), U^d(D)$};

\node at(3,-2.4){$Q_3(A,B) = R^d(A), S^s(A,B), T^d(B)$};
\node at(11,-2.4){$Q_4(A,B,C) = R^d(A,B), S^d(A,C), T^s(B,C)$};

\node at(3,-3){$Q_5(B,C) = R^d(A,B), S^d(A,C), T^s(B,C)$};
\node at(11,-3){$Q_6(A,B) = R^d(A), S^d(A,B), T^d(B,C), U^s(C)$};
\end{scope}

\end{tikzpicture}
    \caption{Examples of queries inside and outside our tractability classes. The static and dynamic relations are adorned with the superscripts $s$ and respectively $d$. In the hypergraphs, there is one node per variable and one hyperedge per relation. Underlined variables are free. Solid (red) hyperedges denote the dynamic relations, the dotted (blue) hyperedges denote the static relations.  
 }
    \label{fig:queries}
\end{figure}

\begin{example}
\rm
Let the query $Q_1(A,B,C) = R^d(A,D), S^d(A,B), T^s(B,C)$ 
in Figure~\ref{fig:queries}. The dynamic relations $R$ and $S$ are adorned with the superscript $d$, while the static relation $T$ is adorned with $s$. The query is not tractable in the all-dynamic setting (as it is not $q$-hierarchical, cf.\@ Section~\ref{sec:prelims}), yet it is in  $\lin$, so it is tractable in the mixed setting. 
\end{example}

In the all-static setting, $\lin$ becomes the class of free-connex $\alpha$-acyclic conjunctive queries, which are those conjunctive queries that admit constant enumeration delay after linear time preprocessing~\cite{BaganDG07}. In the all-dynamic setting, $\lin$ becomes the class of $q$-hierarchical conjunctive queries, which are those conjunctive queries that admit constant enumeration delay and constant update time after linear time preprocessing. Every query  in $\lin$ is free-connex $\alpha$-acyclic and its dynamic sub-query, which is obtained by removing the static relations, is $q$-hierarchical. Yet the queries in  $\lin$ need to satisfy further syntactic constraints on the connections between their static and dynamic relations: For instance, $Q_3$ and $Q_4$ from Figure~\ref{fig:queries} are not in $\lin$ even though they are free-connex $\alpha$-acyclic and their dynamic sub-queries are $q$-hierarchical.

The queries in $\pol \setminus \lin$ are also tractable, yet they require super-linear preprocessing time, unless the Online Matrix-Vector Multiplication or the Boolean Matrix-Matrix Multiplication conjectures fail.  We introduce a new width measure $\fw$ for conjunctive queries, called the {\em preprocessing width}, to characterise the preprocessing time of queries in $\pol$:

\begin{theorem}
\label{thm:polynomial}
    Let a conjunctive query $Q \in \pol$,  the preprocessing width $\fw$ of $Q$, and a database of size $N$. 
    The query $Q$ can be maintained with $\bigO{N^{\fw}}$ preprocessing time, $\bigO{1}$ update time, and $\bigO{1}$ enumeration delay. 
\end{theorem}

Like the queries in $\lin$, every query in $\pol$ also admits a safe rewriting (Section~\ref{sec:evaluation_lin_pol}). 

\begin{example}
\rm
The query $Q_2(A,C,D) = R^d(A,D), S^s(A,B), T^s(B,C), U^d(D)$ from Figure~\ref{fig:queries} is contained in $\pol \setminus\lin$: It is tractable but requires quadratic time preprocessing. The quadratic blowup is due to the creation of a view that is the join of the static relations $S$ and $T$  on the bound variable $B$.
\end{example}

The preprocessing width is not captured by previous width notions, such as the fractional hypertree width ($\fhtw$) of either the static sub-query or of the entire query~\cite{DBLP:journals/talg/Marx10}, as exemplified next. Consider the free-connex $\alpha$-acyclic query $Q_7(A,B,C) = R^s(A,B), S^s(B,C), T^s(A,C),$ $U^d(A,B,C)$, whose $\fhtw$ is 1. Its static sub-query is the triangle join and has $\fhtw=3/2$. The preprocessing width of  $Q_7$ is 1, so it is in $\lin$. The triangle join $Q_8(A,C) = R^s(A,B), S^s(B,C), T^d(A,C)$ has $\fhtw$ $=3/2$ and its static sub-query has $\fhtw$ $=2$. Yet the preprocessing width of $Q_8$ is 2: We  materialise the static sub-query $V^s(A,C) = R^s(A,B), S^s(B,C)$. We may reduce the preprocessing width to 3/2 for the static sub-query by also joining with the dynamic relation $T^d(A,C)$, yet the modified view becomes dynamic and needs to be maintained under updates to $T$. However, this maintenance cannot be achieved with constant update time, while allowing for constant enumeration delay~\cite{Kara:TODS:2020}.

The class $\expo$ contains a broad set of tractable queries that may use up to exponential preprocessing time in the size of the input database:

\begin{theorem}
\label{thm:dichotomy_possible}
Let a conjunctive query $Q$, the static sub-query $\stat(Q)$ of $Q$,  the fractional edge cover number $\rho^*(Q)$ of $Q$, $p = \bigO{N^{\rho^*(\stat(Q))}}$, and a database whose static relations have size $N$. 
If $Q$ is in $\expo$, then it can be maintained with $2^p\cdot p^2$ preprocessing time, $\bigO{1}$ update time, and $\bigO{1}$ enumeration delay. 
\end{theorem}

\nop{
    Let a CQ $Q$ and a database of size $N$. 
\item If $Q$ is in $\expo$, then it can be evaluated with $2^p$ time preprocessing, 
$\bigO{1}$ update time, and $\bigO{1}$ enumeration delay, where $p = \bigO{N^{\rho^*(\stat(Q))}}$ and $\rho^*(\stat(Q))$ is the fractional edge cover number of the static sub-query $\stat(Q)$.

\item If $Q$ is not in $\expo$ and \change{does not have repeating relation symbols},  then there is no $\gamma > 0$ such that $Q$ can be evaluated with arbitrary preprocessing time, $\bigO{N^{1/2-\gamma}}$ update time, and $\bigO{N^{1/2-\gamma}}$ enumeration delay, unless the Online Matrix-Vector Multiplication conjecture fails. 
}

The class $\expo$ is merely of theoretical interest since it comes with  preprocessing time that is exponential in the size of the input database. The reason for this high preprocessing time is to ensure constant enumeration delay; if we would allow the enumeration delay to become linear, then the preprocessing time would collapse to linear for the $\alpha$-acyclic conjunctive queries in $\expo$. 
The fractional edge cover number characterises the worst-case optimal result size of the static sub-query of $Q$~\cite{AtseriasGM13}.

\nop{The possible results of $Q_3$ are completely predetermined by its static sub-query: 
They are the subsets of the relation $S$, while the updates to the dynamic relations only act as selectors in this powerset.
 Such \change{CQs} require exponential time preprocessing to index the $2^N$ many possible subsets of the relation $S^s$ of size $N$.}

\begin{example}
\rm
The query $Q_3(A,B) = R^d(A), S^s(A,B), T^d(B)$  from Figure~\ref{fig:queries} is in $\expo \setminus \pol$. This query can be maintained with constant update time and enumeration delay at the expense of exponential preprocessing time: The possible results of $Q_3$ are any of the $2^N$ possible subsets of the static relation $S^s$ of size $N$, while the updates to the dynamic relations only act as selectors in this powerset.
The dynamic sub-query of $Q_3$, which is defined as the sub-query over the dynamic relations, is $q$-hierarchical. The static sub-query of $Q_3$, which is defined as the sub-query over the only static relation, is trivially free-connex $\alpha$-acyclic. This means that, {\em when taken in isolation}, the dynamic sub-query can be evaluated with constant update time and enumeration delay after linear time preprocessing, while the static sub-query can be evaluated with constant enumeration delay after linear time preprocessing. Yet $Q_3$ is not in $\lin$: It does not admit constant update time and enumeration delay {\em after linear time preprocessing}.
As we show in Section~\ref{sec:outlook}, the queries $Q_4$, $Q_5$, and $Q_6$ from Figure~\ref{fig:queries} are not in $\expo$, and we do not know whether they are tractable.
\end{example}

The queries in $\expo \setminus \pol$ may not admit safe rewritings that rely solely on joins and projections. Take again $Q_3(A,B) = R^d(A), S^s(A,B), T^d(B)$ from Figure~\ref{fig:queries}. There is no safe rewriting of this query that solely relies on projections and joins. Any rewriting that supports constant-delay enumeration of the query result must contain a 
materialised view that either joins: 
(i) $R^d$ and  $S^s$; 
(ii) $S^s$ and $T^d$; or
(iii) $R^d$ and  $T^d$.
None of these views can be maintained with constant update time. Consider the view $V(A,B)= R^d(A), S^s(A,B)$. An insert of a value $a$ to $R$ requires to iterate over all $B$-values paired with $a$ in $S^s$ in order to propagate the change to the view $V$. The number of such $B$-values can be linear. Hence,  the update time is linear.
Section~\ref{sec:evaluation_expo} sketches our evaluation strategy for the queries in $\expo$.

The proofs of formal statements are deferred to the appendix.

%% file: prelims.tex
\section{Preliminaries}
\label{sec:prelims}
For a natural number $n$, we define $[n] = \{1,2,\ldots,n\}$. In case $n=0$, we have $[n]=\emptyset$.

\paragraph*{Databases with Static and Dynamic Relations}

Following standard terminology, a relation is a finite set of tuples and a database is a
finite set of relations~\cite{AbiteboulHV95}. The size $|R|$ of a relation $R$ is the number of its tuples. The size $|D|$ of a database $D$ is given by the sum of the sizes of its relations.
The relation $R$ is {\em dynamic} if it is subject to updates;
otherwise, it is {\em static}. 
To emphasize that $R$ is static or dynamic, we write $R^s$ or respectively $R^d$.

\paragraph*{Conjunctive Queries}
A conjunctive query (CQ or query for short) has the form  
\begin{align*}
    Q(\bm X) = R_1^d(\bm X_1), \ldots ,R^d_k(\bm X_k),
    S_1^s(\bm Y_1), \ldots , S^s_{\ell}(\bm Y_{\ell})
    \label{eq:join-query}
\end{align*}
where 
$Q$, $R_i$, and $S_j$ are {\em relation symbols}; 
$R_i^d(\bm X_i)$ are dynamic body atoms; 
$S_j^s(\bm Y_j)$ are static body atoms; $Q(\bm X)$ is the head atom; 
$\vars(Q) = \bigcup_{i \in [k]} \bm X_i \cup \bigcup_{j \in [\ell]} \bm Y_j$ is the set of variables of $Q$;
$\free(Q) = \bm X \subseteq \vars(Q)$ is the set of {\em free} variables; 
$\vars(Q)\setminus\free(Q)$ is the set of {\em bound} variables; 
$\atoms(Q)$ is the set of body atoms;
$\atoms(X)$ is the set of body atoms with variable $X$.
The static (dynamic) sub-query $\stat(Q)$ ($\dyn(Q)$) is obtained from $Q$ as follows: Its body is the conjunction of all static (dynamic) atoms of $Q$ and its free variables are the free variables of $Q$ that appear in static (dynamic) atoms of $Q$. We visualise queries as hypergraphs where nodes are query variables (with the free variables underlined), solid red hyperedges represent dynamic atoms, and dotted blue hyperedges represent static atoms.

\begin{example}
\rm
Consider the query 
$Q_2(A,C,D) = R^d(A,D), S^s(A,B), T^s(B,C), U^d(D)$
and its hypergraph
from Figure~\ref{fig:queries}.
Its static and dynamic sub-queries are 
$Q^s(A,C)$ $=$ $S^s(A,B),$ $T^s(B,C)$
and respectively 
$Q^d(A,D) = R^d(A,D), U^d(D)$.
\end{example}

A query  is {\em $\alpha$-acyclic} if we can construct a tree, called 
join tree, such that: 
(1) the nodes of the tree are the body atoms of the query, and 
(2) for each variable, the following holds: if the variable appears in two body atoms, then it appears in all 
body atoms on the unique path between these two body atoms in the tree~\cite{Brault-Baron16}.
A query is {\em free-connex $\alpha$-acyclic} if it is $\alpha$-acyclic and remains $\alpha$-acyclic after adding the head atom 
$Q(free(Q))$ to its body~\cite{BraultPhD13}.

A query is {\em hierarchical} if for any two variables $X$ and $Y$, it holds that $\atoms(X) \subseteq \atoms(Y)$,
$\atoms(Y) \subseteq \atoms(X)$, or $\atoms(X) \cap \atoms(Y) = \emptyset$~\cite{Suciu:PDB:11}. 
A query is $q$-hierarchical if it is hierarchical and for any two variables $X$ and $Y$ with 
$\atoms(X) \supsetneq \atoms(Y)$, it holds that: if $Y$ is free, 
then $X$ is free~\cite{BerkholzKS17}.

A {\em path} in a query $Q$ is a sequence $(X_1, \ldots, X_n)$  of variables in $Q$ such that each variable appears at most once in the sequence and each two adjacent variables $X_i$ and $X_{i+1}$ are contained in a 
body atom from $Q$, $\forall i\in[n-1]$.
We use set operations on paths as on  sets of their variables. Two variables $X_1$ and $X_n$ are \textit{connected} if there is a path $(X_1, \ldots, X_n)$. 
Two atoms $R(\bm Y)$ and $S(\bm Z)$ are {\em connected} if there are connected variables $X_1 \in \bm Y$ and $X_n \in \bm Z$.
A set $\calC$ of body atoms in $Q$ is a {\em connected component} of $Q$ if any two atoms in $\calC$ are connected and this does not hold for any superset of $\calC$.

\begin{example}
\rm
The query $Q_1(A,B,C) = R^d(A,D), S^d(A,B), T^s(B,C)$ from Figure~\ref{fig:queries} has the path $(D, A, B, C)$ that connects (i) the variables $D$ and $C$ and (ii) the atoms $R^d(A,D)$ and $T^s(B,C)$.
\end{example}

\paragraph*{Dynamic Query Evaluation}

The problem of dynamic query evaluation is as follows: 
Given a query $Q$ and a database with static and dynamic relations, the problem is to maintain the query result under updates to the dynamic relations and to enumerate the tuples in the query result.

A single-tuple update to a relation $R$ is an insert of a tuple into $R$ or a delete of a tuple from $R$.
We denote the insert of a tuple $\bm t$ by $+\bm t$ and the delete of $\bm t$ by $-\bm t$. 
In this paper, we consider the set semantics, where a tuple is in or out of the database; in particular, we do not consider the setting, where each tuple is associated with a multiplicity that may be positive or negative as in prior work, e.g.,~\cite{DBT:VLDBJ:2014,Kara:TODS:2020}.

Following prior work, e.g., \cite{BerkholzKS17, Kara:TODS:2020, FIVM:VLDBJ:2023}, we decompose the time complexity of dynamic query evaluation into {\em preprocessing time}, {\em update time}, and {\em enumeration delay}. The preprocessing time is the time to compute a data structure that represents the query result before receiving any update. 
The update time is the time to update the input database and data structure under a single-tuple insert or delete.
The enumeration delay is the maximum of three times: the time between the start of the enumeration process and outputting the first tuple, the time between outputting any two consecutive tuples, and the time between outputting the last tuple and the end of the enumeration process~\cite{DurandFO07}.

\paragraph*{Computation Model and Data Complexity}

\nop{ 
To address the elements in a value set with a size that is polynomial in the input size $N$, 
we use an index with pointers of size $\bigO{\log N}$, allowing for constant-time lookups.
To address the elements in a value set whose size is polynomial 
in the input size $N$, we use  an index whose pointers are of size $\bigO{\log N}$ and that admits constant lookup time.}

We consider the RAM model of computation. 
We assume that each relation  $R$ is implemented by a data structure 
of size $O(|R|)$ that can: 
(1) look up, insert, and delete entries in constant time, and
(2) enumerate all stored entries in $R$ with constant delay.
For a set $\bm S {\subsetneq} \bm X$,
where $\bm X$ is the set of attributes of R, 
we use an index data structure that, for any tuple
$\tup{t}$ over the attributes in $\bm S$, can: 
(3) enumerate all tuples in $\sigma_{\bm S=\tup{t}}R$ with constant delay, and 
(4) insert and delete index entries in constant time. 
We further require the following assumption in Section~\ref{sec:evaluation_expo}: Given a relation $T$, whose size is polynomial in the input database size $N$, we can lookup for any tuple in $T$ in constant time.

We report time complexities as functions of the database size $N$ 
under data complexity (so the query is fixed and has constant size).
Constant update time and constant delay therefore mean that they do not depend on the database size.


\paragraph*{Fractional Edge Cover Number}

\begin{definition}[Fractional Edge Cover~\cite{AtseriasGM13}]
\label{def:fractional_edge_cover}
Given a conjunctive query $Q$ 
and a set 
$\bm F \subseteq \vars(Q)$,    
a {\em fractional edge cover}
of $Q$ for $\bm F$ is a solution
$\boldsymbol{\lambda} = (\lambda_{R(\bm X)})_{R(\bm X) \in \atoms(Q)}$ to the following 
linear program: 
\begin{align*}
\text{minimize} & \TAB\sum_{R(\bm X) \in\, \atoms(Q)} \lambda_{R(\bm X)} && \\[3pt]
\text{subject to} &\TAB \sum_{R(\bm X)\in\, \atoms(Q) \text{ s.t. } X \in \bm X}\hspace{-0.5cm} \lambda_{R(\bm X)} \geq 1 && \text{ for all } X \in \bm F \text{ and } \\[3pt]
& \TAB\lambda_{R(\bm X)} \in [0,1] && \text{ for all } R(\bm X) \in \atoms(Q)
\end{align*}
The optimal objective value of the above program 
is called the {\em fractional edge cover number} of $Q$ for the variable set $\bm F$ 
and denoted as $\rho_{Q}^{\ast}(\bm F)$.  
We abbreviate $\rho_{Q}^{\ast}(\vars(Q))$  by $\rho^*(Q)$.
\end{definition}
\nop{If $Q$ is clear from the context, we omit 
the index $Q$ and write $\rho^{\ast}(\bm F)$.} 
For a query $Q$ without bound variables, a set $\bm F \subseteq \vars(Q)$,
and a database of size $N$, $N^{e}$ with $e = \rho_{Q}^*(\bm F)$ is an upper bound on the worst-case size of the  
result of $Q$ projected onto the variables in $\bm F$~\cite{AtseriasGM13}.
The result of $Q$ can be computed in time 
$\bigO{N^{\rho^*(Q)}}$~\cite{Ngo:JACM:18}.

\paragraph*{Complexity-Theoretic Conjectures}
In this work, we establish complexity lower bounds based on the following widely accepted complexity-theoretic conjectures.

\begin{definition}[Online Matrix-Vector Multiplication (OMv) Problem~\cite{Henzinger:OMv:2015}]
\label{def:OMV_problem}
We are given an $n \times n$ Boolean matrix $M$ and receive $n$ Boolean column vectors $v_1,\dots, v_n$ of size $n$, one by one. After seeing each vector $v_i$, the task is to output the multiplication $Mv_i$ before seeing the next vector.
\end{definition}

\begin{conjecture}[OMv Conjecture~\cite{Henzinger:OMv:2015}]
\label{conj:OMV}
For any constant $\gamma > 0$, there is no algorithm that solves the 
OMv problem in time $\bigO{n^{3-\gamma}}$.
\end{conjecture}

\begin{definition}[Online Vector-Matrix-Vector Multiplication (OuMv) Problem~\cite{Henzinger:OMv:2015}]
\label{def:OuMV_problem}
We are given an $n \times n$ Boolean matrix $M$ and receive $n$ pairs of Boolean column vectors $(u_1, v_1),\dots, (u_n, v_n)$ of size $n$, one by one. 
After seeing each pair of vectors $(u_i, v_i)$, the task is to output the multiplication $u_iMv_i$ before seeing the next pair.
\end{definition}

The following OuMv conjecture is implied by the OMv conjecture.

\begin{conjecture}[OuMv Conjecture~\cite{Henzinger:OMv:2015}]
\label{conj:OuMV}
For any constant $\gamma > 0$, there is no algorithm that solves OuMv 
problem in time $\bigO{n^{3-\gamma}}$.
\end{conjecture}

%% file: rewriting.tex
\section{Safe Query Rewriting}
\label{sec:rewriting}

A rewriting of a query is a project-join plan for the query. In the context of dynamic query evaluation, query rewritings have been previously used under the term {\em view trees} due to their natural tree-shaped graphical depiction~\cite{FIVM:VLDBJ:2023, DBLP:journals/lmcs/KaraNOZ23}.
In this paper, we use so-called {\em safe} query rewritings, which ensure tractable dynamic evaluation.

A {\em rewriting} of a conjunctive query $Q$ using views (simply, rewriting of $Q$) is a forest 
$\calT = \{T_i\}_{i \in [n]}$ of trees $T_i$ 
where the leaves are the body atoms of $Q$ and each inner node is a view $V$ such that:

\begin{itemize}
\item If $V$ has a single child, then $V$ is a projection of the child onto some variables; we call $V$ a  {\em projection view}. 
\item If $V$ has several children, then $V$ is the join of its children; we call $V$ a {\em join view}. 
\end{itemize}
A view $V$ is  {\em dynamic} if the subtree rooted at $V$ contains a dynamic atom. For convenience, we also refer to the atoms in a rewriting as views.

\begin{example}
\rm
 Figure~\ref{fig:view_rewritings} gives three (out of many possible) rewritings of the query 
 $Q_2(A,C,D)$ $=$ $R^d(A,D),$ $S^s(A,B),$ $T^s(B,C),U^d(D)$
 from Figure~\ref{fig:queries}.
 Each rewriting is depicted as a tree. 
 In all rewritings, the view $V_{ST}$ is static, while the view $V_{RST}$ is dynamic, since it contains the dynamic relation $R^d$ in its subtree. 
\end{example}

\begin{figure}[t]
    \centering
\begin{tikzpicture}
\hspace{-0.5cm}
 \begin{scope}[xshift=0cm, yshift= 0cm]
    \node[draw=none,fill=none](RSTU) at (-2, 3.9) {$V_{RSTU}(\underline{A},B, 
    \underline{C}, \underline{D})$};
    \node[draw=none,fill=none](U) at (-3.2, 2.6) {$U^d(\underline{D})$};
    \node[draw=none,fill=none](RST) at (-1, 2.6) {$V_{RST}(\underline{A},B, 
    \underline{C}, \underline{D})$};
    \node[draw=none,fill=none](R) at (-2.2, 1.3) {$R^d(\underline{A},\underline{D})$};
    \node[draw=none,fill=none](ST) at (0, 1.3) {$V_{ST}(\underline{A},B,\underline{C})$};
    \node[draw=none,fill=none](S) at (-0.8, 0) {$S^s(\underline{A}, B)$};
    \node[draw=none,fill=none](T) at (0.8, 0) {$T^s(B, \underline{C})$};    
    
        \draw (S) -- (ST);
    \draw (T) -- (ST);
    \draw (R) -- (RST);
    \draw (ST) -- (RST);    
    \draw (U) -- (RSTU);        
    \draw (RST) -- (RSTU);                    
  \end{scope}   
 
 \begin{scope}[xshift=5cm, yshift= 0cm]
    \node[draw=none,fill=none](RSTU) at (-1.8, 4.4) {$V_{RSTU}(\underline{A}, \underline{C},\underline{D})$};
    \node[draw=none,fill=none](R) at (-2.8, 3.3) {$U^d(\underline{D})$};
    \node[draw=none,fill=none](STU) at (-0.9, 3.3) {$V_{RST}(\underline{A},\underline{C},\underline{D})$};
    \node[draw=none,fill=none](S) at (-1.8, 2.2) {$R^d(\underline{A},\underline{D})$};
    \node[draw=none,fill=none](TU') at (0, 2.2) {$V_{ST}'(\underline{A},\underline{C})$};
    \node[draw=none,fill=none](TU) at (0, 1.1) {$V_{ST}(\underline{A},B,\underline{C})$};
    \node[draw=none,fill=none](T) at (-0.8, 0) {$S^s(\underline{A},B)$};
    \node[draw=none,fill=none](U) at (0.8, 0) {$T^s(B, \underline{C})$};

    \draw (T) -- (TU);
    \draw (U) -- (TU);
    \draw (TU) -- (TU');
    \draw (S) -- (STU);
    \draw (TU') -- (STU);    
    \draw (STU) -- (RSTU);    
    \draw (R) -- (RSTU);        
  \end{scope}   

 \begin{scope}[xshift = 10cm, yshift= 0cm]
    \node[draw=none,fill=none](RSTU) at (-1.8, 4.8) {$V_{RSTU}(\underline{D})$};
    \node[draw=none,fill=none](STU') at (-0.9, 4) {$V_{RST}'(\underline{D})$};
    \node[draw=none,fill=none](R) at (-2.7, 4) {$U^d(\underline{D})$};
    \node[draw=none,fill=none](STU) at (-0.9, 3.2) {$V_{RST}(\underline{A},\underline{D})$};
    \node[draw=none,fill=none](S) at (-1.8, 2.4) {$R^d(\underline{A},\underline{D})$};
    \node[draw=none,fill=none](TU'') at (0, 2.4) {$V_{ST}''(\underline{A})$};
    \node[draw=none,fill=none](TU') at (0, 1.6) {$V_{ST}'(\underline{A},\underline{C})$};
    \node[draw=none,fill=none](TU) at (0, 0.8) {$V_{ST}(\underline{A},B,\underline{C})$};
    \node[draw=none,fill=none](T) at (-0.8, 0) {$S^s(\underline{A},B)$};
    \node[draw=none,fill=none](U) at (0.8, 0) {$T^s(B, \underline{C})$};    
    
    \draw (T) -- (TU);
    \draw (U) -- (TU);
    \draw (TU) -- (TU');
    \draw (TU') -- (TU'');
    \draw (S) -- (STU);
    \draw (TU'') -- (STU);    
    \draw (STU) -- (STU');    
    \draw (R) -- (RSTU);        
    \draw (STU') -- (RSTU);          
  \end{scope}   

\end{tikzpicture}
    \caption{Three rewritings of the query 
    $Q_2(A,C,D) = R^d(A,D), S^s(A,B), T^s(B,C),U^d(D)$ from Figure~\ref{fig:queries}.
    The first two rewritings are not safe, while the last one is safe. }
    \label{fig:view_rewritings}
\end{figure}

Next, we define {\em safe} query rewritings, which have 
four properties. The two {\em correctness} properties ensure that the views correctly encode the query result as a factorised representation~\cite{OlteanuS:SIGREC16}.
The {\em update} property guarantees that the dynamic views can be maintained in constant time under single-tuple updates to any dynamic relation. 
The {\em enumeration} property ensures that the tuples in the query result can be listed from the views with constant delay.

 \begin{definition}[Safe Query Rewriting]
 \label{def:safe_rewriting}
A rewriting $\calT = \{T_i\}_{i \in [n]}$ 
is {\em safe} for a conjunctive query $Q$ if: 
\begin{description}
\item[Correctness]
    (1) For each connected component $\calC$ of $Q$, there is a tree  $T_i$ that contains all atoms in $\calC$. 
    (2) For any projection view $V'(\bm Y)$ with  child view $V(\bm X)$, it holds that
    each atom of $Q$ containing a variable from $\bm X \setminus \bm Y$ is contained only in the subtree rooted at $V$.

\item[Update] The set of free variables of each dynamic view includes the 
set of free variables of each of its sibling views. 

\item[Enumeration]
Each tree $T_i$ has a connected set $\mathcal{V}_i$ of views  in $T_i$ that contains the root of $T_i$ such that
$\bigcup_{i\in [n]} \bigcup_{V(\bm X) \in \mathcal{V}_i} \bm X = \free(Q).$
\end{description}
 \end{definition}

 The {\em materialisation time} for a query rewriting is the time 
 to materialise all its views (including both static and dynamic views).

 \begin{proposition}
 \label{prop:safe_rewriting_complexity}
Let a conjunctive query $Q$ and a database of size $N$.
If $Q$ has a safe rewriting with $f(N)$ materialisation time for some function $f$, then $Q$ can be evaluated with $f(N)$ preprocessing time, $\bigO{1}$ update time, and $\bigO{1}$ enumeration delay. 
  \end{proposition}

\begin{example}
\rm
The first rewriting in Figure~\ref{fig:view_rewritings} is not safe: 
It violates the enumeration property because the root view $V_{RSTU}$ contains the bound variable $B$; thus there is no guarantee of reporting distinct $(A,C,D)$-values with constant delay.
One possibility is to project away $B$ before starting the enumeration but this requires time linear in the size of $V_{RSTU}$. 
The rewriting also violates the update property: for instance, the set of free variables of the dynamic view $R^d(A,D)$ does not  include the set of free  variables of its sibling $V_{ST}(A,B,C)$, thus 
computing the change in $V_{RST}$ for an update to $R^d$
requires iterating over linearly many $(B,C)$-values from $V_{ST}$ that are paired with the $A$-value from the update.

The second rewriting is also not safe: 
It satisfies the enumeration property as the root view encodes the query result 
but fails on the update property for both dynamic atoms.

The third rewriting is safe and admits $\bigO{1}$ update time and $\bigO{1}$ enumeration delay, per Proposition~\ref{prop:safe_rewriting_complexity}. We next show how to handle updates and enumerate from this rewriting.

We can propagate an update to $R^d$ or $U^d$ to their ancestor views in constant time. 
Consider an insert of a new tuple $(a,d)$ to relation $R^d$.
To compute the change in $V_{RST}$,
we check if $a$ exists in $V_{ST}''$ via a constant-time lookup.
If not, we stop as no further propagation is needed; 
otherwise, we insert $(a,d)$ into $V_{RST}$ in constant time.
We maintain $V_{RST}'$ by inserting $d$ into that view.
We compute the change in the root $V_{RSTU}$ by checking if $d$ exists in $U^d$ via a constant-time lookup, and if so, we insert $d$ into the root. 
Propagating an insert to $U^d$ requires a lookup in $V_{RST}'$ and an insert into $V_{RSTU}$, both in constant time.
Deletes to $R^d$ and $U^d$ are handled analogously. 
Thus, updates in this rewriting take constant time.

We can enumerate the distinct 
tuples in the query result using three nested for-loops. 
The first loop iterates over the $D$-values in $V_{RSTU}$; 
the second loop iterates over the $A$-values in $V_{RST}$ paired with a $D$-value from the first loop;
and the third loop iterates over the $C$-values in $V'_{ST}$ paired with an $A$-value from the second loop. 
Hence, each distinct  result tuple can be constructed in constant time. 

The time to compute the view $V_{ST}$ is quadratic because in the worst case each
$A$-value in $S^s$ is paired with each $C$-value in $T^s$.
All other views in the rewriting are either projection views or
semi-joins of one child view with another child view. 
Thus, the overall computation time for the rewriting is $\bigO{N^2}$,
where $N$ is the database size.   
\end{example}

%% file: query_classes.tex
\section{New Query Classes}
\label{sec:query_classes}
In this section we introduce the query classes $\lin$, $\pol$, and 
$\expo$. We first define syntactic properties of the classes $\lin$ and $\pol$ that guarantee 
the existence of safe rewritings using views. The class $\expo$ contains queries that do not satisfy these properties.

\begin{definition}
A path $\bm P$ connecting two atoms $R(\bm X)$ and $S(\bm Y)$ in a conjunctive query $Q$ is {\em safe} if $\bm P\cap \bm X\cap \bm Y \neq\emptyset$. In particular, the path is {\em body-safe} 
if the two atoms are body atoms and it is {\em head-safe} 
if one atom is a body atom and the other atom is the head atom.
A conjunctive query $Q$ is {\em well-behaved} if: (1) 
all paths connecting dynamic body atoms are body-safe; and (2) 
all paths connecting a dynamic body atom with the head atom are head-safe.
\end{definition}

\begin{example}
\label{ex:non-safe}
\rm
The queries $Q_3$--$Q_6$ from Figure~\ref{fig:queries}
are not well-behaved.  
The path $(A,B)$
in $Q_3(A,B, C)$ $=$ $R^d(A),$ $S^s(A,B),$ $T^d(B)$
connects the two dynamic body atoms $R^d(A)$ and $T^d(B)$,
but it is not body-safe, since 
$\{A,B\} \cap \{A\} \cap \{B\} = \emptyset$.
The path $(B,C)$
in  
$Q_4(A,B, C)$ $=$ $R^d(A, B),$ $S^d(A, C),$ $T^s(B, C)$ 
connects
the two dynamic body atoms $R^d(A,B)$ and $S^d(A,C)$,
but the path is not body-safe, since  
 $\{B,C\} \cap \{A,B\} \cap \{A,C\} = \emptyset$. 
The path $(A,B)$ in  
$Q_5(B, C)$ $=$ $R^d(A, B),$ $S^d(A, C),$ $T^s(B, C)$ 
connects
the dynamic body atom $S^d(A,C)$ with the head atom $Q_5(B, C)$,
but it is not head-safe, since
$\{A,B\} \cap \{A,C\} \cap \{B,C\} = \emptyset$. 
The path $(A,B)$
in $Q_6(A, B) = R^d(A), S^d(A, B), T^d(B, C), U^s(C)$
connects the two dynamic body atoms $R^d(A)$ and $T^d(B, C)$,
but it is not body-safe, since 
$\{A,B\} \cap \{A\} \cap \{B,C\} = \emptyset$.
\nop{
The query $Q_5(B, C) = R^d(A, B), S^d(A, C), T^s(B, C)$
from Figure~\ref{fig:queries} is not well-behaved since the path $(A,B)$ connects
the dynamic body atom $S^d(A, C)$ with the head atom $Q_5(B,C)$,  
yet $\{A,B\} \cap \{A,C\} \cap \{B,C\} = \emptyset$.
}
\end{example}

We can check efficiently whether a query is well-behaved.

\begin{proposition}
\label{prop:safety_properites_poly_time_check}
For any conjunctive query $Q$ with $n$ variables and $m$ atoms, we can decide in $O(n^2 \cdot m^2)$ time whether $Q$ is well-behaved.
\end{proposition}

The well-behavedness property of a query implies the $q$-hierarchical property of its dynamic sub-query.

\begin{proposition}
\label{prop:safety_properties_q-hierarchical}
\begin{itemize}
\item Any conjunctive query without static relations is q-hierarchical if and only if it
is well-behaved. 
\item The dynamic sub-query of any well-behaved conjunctive query is q-hierarchical. 
\end{itemize}
\end{proposition}

In Section \ref{sec:evaluation_lin_pol}, we show that every well-behaved
query admits a safe rewriting, that is, it can be rewritten into a forest of view trees that support constant update time and constant enumeration delay. To obtain linear preprocessing time, we further require that the query is free-connex 
$\alpha$-acyclic.

\begin{definition}[$\lin$]
\label{def:class_lin}
A conjunctive query is in $\lin$ if it is free-connex $\alpha$-acyclic and well-behaved.  
\end{definition}

For the class $\pol$, we skip the condition that the query is 
free-connex $\alpha$-acyclic. 

\begin{definition}[$\pol$]
\label{def:class_pol}
A conjunctive query is in $\pol$ if it is well-behaved.  
\end{definition}

\begin{example}
\label{ex:well-behaved}
\rm     
    Consider the queries from Figure~\ref{fig:queries}.    
    The query $Q_1$ is in $\lin$, since it is free-connex $\alpha$-acyclic and well-behaved.
    The query $Q_2$ is well-behaved but not free-connex $\alpha$-acyclic, since adding its head atom $Q_2(A,C,D)$ to its body yields a cyclic query. Hence, $Q_2$ is in $\pol$.
    The queries $Q_3$-$Q_6$ are not in $\pol$, since they are not well-behaved, as explained in 
    Example~\ref{ex:non-safe}.
\end{example}
We now define our most permissive class of tractable queries:

\begin{definition}[$\expo$]
\label{def:class_exponential}
A conjunctive query is in $\expo$ if it is well-behaved or every variable occurring in a dynamic atom also 
occurs in a static atom.
\end{definition}

\begin{example}
\rm
The query $Q_3(A,B) = R^d(A), S^s(A,B), T^d(B)$ in Figure~\ref{fig:queries} 
is not well-behaved   
but is in $\expo$, since the variables of the two dynamic atoms
$R^d(A)$ and $T^d(B)$  
also occur in the static atom $S^s(A,B)$.

The queries $Q_4$--$Q_6$ are not contained in $\expo$, since they are not well-behaved 
(as explained in Example~\ref{ex:non-safe}) and have variables in dynamic atoms that do not occur in static atoms:
The variable $A$ in $Q_4$ and $Q_5$ does not appear in the only static atom $S^d(B,C)$, and the 
variables $A$ and $B$ in $Q_6$ do not appear in the only static atom $U^s(C)$.
\end{example}

%% file: evaluation_lin_pol.tex
\section{Evaluation of Queries in $\lin$ and $\pol$}
\label{sec:evaluation_lin_pol}
This section introduces our evaluation strategy for the $\pol$ and $\lin$ queries.
Section~\ref{sec:variable_orders} introduces variable orders, 
which guide the construction of safe rewritings for such 
queries, and the preprocessing width $\fw$ of a query 
based on its variable orders.
Section~\ref{sec:safe_rewriting_lin_pol} shows that the construction of a
safe rewriting for queries in $\pol$ can be done in $\bigO{N^\fw}$ time, where $N$ is the size of the input database.

\subsection{Variable Orders}
\label{sec:variable_orders}
We say that two variables in a query are {\em dependent} if 
they appear in the same body atom. 
A {\em variable order} (VO) for a query $Q$ in $\pol$ is a 
forest $\omega = \{\omega_i\}_{i \in [n]}$ of trees 
$\omega_i$ such that: 
(1) There is a one-to-one mapping between the nodes of $\omega$ and  the variables in $Q$;
(2) the variables of each body atom in $Q$ lie on a root-to-leaf path in 
$\omega$~\cite{OlteanuZ15,OlteanuS:SIGREC16}. 
We denote by $\vars(\omega)$ the set of variables in $\omega$ and by
$\omega_X$ the subtree of $\omega$ rooted at $X$.
We associate each VO $\omega$ with a dependency 
function $\dep_{\omega}$ that maps each variable $X$ in $\omega$ 
to the subset of its ancestors on which the variables in $\omega_X$ 
depend: 
$\dep_{\omega}(X) = \{Y | Y \text{ is an ancestor of } X \text{ and } 
\exists Z \in \omega_X \text{ s.t. $Y$ and $Z$ are dependent}\}$. 
We further extend a VO by adding each body atom of the query as the child of the lowest variable in the VO that belongs to the atom.

We adapt the notion of canonical VOs from the all-dynamic setting~\cite{DBLP:journals/lmcs/KaraNOZ23}
to the setting with both static and dynamic relations. 
A VO is {\em canonical} if the set of variables of each dynamic body atom $A$ is 
 the set of inner nodes of the root-to-leaf path that ends with $A$.
It is {\em free-top} if no bound variable is an ancestor of a free variable. 
It is  {\em well-structured}  if it is canonical and free-top.
We denote by $\WVO(Q)$ the set of well-structured VOs of $Q$.

\begin{figure}[t]
\begin{minipage}{0.2\textwidth}
\begin{tikzpicture}
    \node[draw=none,fill=none](a) at (1, -1) {$\underline{A}$};
    \node[draw=none,fill=none](b) at (0, -2) {$\underline{B}$};
    \node[draw=none,fill=none](c) at (2, -2) {$C$};
    \node[draw=none,fill=none](d) at (1, -2.5) {$D$};
    \draw[rotate=45, red] (-0.65,-1.4) ellipse (1.3cm and 0.4cm);
    \draw[rotate=-45, red] (2.05, 0) ellipse (1.3cm and 0.4cm);
    \draw[blue, dotted , line width = 1.2] (1,-1.75) ellipse (.4cm and 1.2cm);
    \draw[rotate=30,blue, dotted , line width = 1.2] (.2,-2.7) ellipse (1cm and .4cm);
\end{tikzpicture}
\end{minipage}
\quad
\begin{minipage}{0.4\textwidth}
\begin{tikzpicture}
    \node[draw=none,fill=none](a) at (0, -1) {$\underline{A}$};
    \node[draw=none,fill=none](b) at (-1, -2) {$\underline{B}$};
    \node[draw=none,fill=none](c) at (1, -2) {$C$};
    \node[draw=none,fill=none](d) at (1, -3) {$D$};

    \node[draw=none,fill=none](R) at (-1, -3) 
    {\color{red} $R^d(\underline{A}, \underline{B})$};

    \node[draw=none,fill=none,anchor=west](S) at (1.7, -3) 
    {\color{red} $S^d(\underline{A}, C)$};

    \node[draw=none,fill=none,anchor=east](Y) at (0.7, -4) 
    {\color{blue} $Y^s(\underline{A}, D)$};

    \node[draw=none,fill=none,anchor=west](Z) at (1.3, -4) 
    {\color{blue} $Z^s(C, D)$};

    \draw (b) -- (a);
    \draw(c) -- (a);
    \draw (c) -- (d);
    \draw (c) -- (S);
    \draw (b) -- (R);
    \draw (d) -- (Y);
    \draw (d) -- (Z);

    \node[anchor=west] at (1, 0) {\footnotesize $\dep_\omega(A) = \emptyset$};
    \node[anchor=west] at (1, -0.5) {\footnotesize $\dep_\omega(B) = \dep_\omega(C) = \{A\}$};
    \node[anchor=west] at (1, -1) {\footnotesize $\dep_\omega(D) = \{A,C\}$};        
\end{tikzpicture}
\end{minipage}
\qquad
\begin{minipage}{0.35\textwidth}
\begin{tikzpicture}
\node(VA) at (.7, -1.4) 
{\color{red} $V_{A}(\underline{A})$};

\node(VB) at (-.3, -2.2) 
{\color{red} $V'_{B}(\underline{A})$};

\node(VF1) at (1.8, -2.2) 
{\color{red} $V'_{C}(\underline{A})$};

\node(VC) at (1.8, -3) 
{\color{red} $V_{C}(\underline{A}, C)$};

\node(R) at (-.3, -3) 
{\color{red} $R^d(\underline{A}, \underline{B})$};

\node(S) at (2.7, -3.8) 
{\color{red} $S^d(\underline{A}, C)$};

\node(VD1) at (0.9, -3.8) 
{\color{blue} $V'_{D}(\underline{A}, C)$};

\node(VD) at (0.9, -4.6) 
{\color{blue} $V_{D}(\underline{A}, C, D)$};

\node(Y) at (0, -5.4) 
{\color{blue} $Y^s(\underline{A}, D)$};

\node(Z) at (1.8, -5.4) {
\color{blue} $Z^s(C, D)$};

\draw (VB) -- (VA);
\draw (VB) -- (R);
\draw (VF1) -- (VA);
\draw (VF1) -- (VC);
\draw (S) -- (VC);
\draw (VD1) -- (VC);
\draw (VD1) -- (VD);
\draw (VD) -- (Y);
\draw (VD) -- (Z);
\end{tikzpicture}
\end{minipage}
\caption{
(left to right) The hypergraph of $Q(A,B) = R^d(A,B), S^d(A,C), Y^s(A,D), Z^s(C,D)$, a 
well-structured VO $\omega$ for $Q$, and the view tree for $Q$ constructed by the procedure $\rewrite(\omega)$ from Figure~\ref{fig:view_tree_construction}. Dynamic views are in red, static views are in blue.
The query $Q$ is well-behaved and the preprocessing width is 2, thus $Q \in \pol$.
}
    \label{fig:pol_class_construction}
\end{figure}

    \begin{example}
    \rm
      The query $Q$ from Figure~\ref{fig:pol_class_construction} is in $\pol$ as it 
      is well-behaved. It is not in $\lin$ as it is not free-connex $\alpha$-acyclic. 
      The dynamic sub-query $Q^d(A,B) = R^d(A,B),S^d(A,C)$ is q-hierarchical, per Proposition~\ref{prop:safety_properties_q-hierarchical}.
      The figure shows a well-structured VO for the query. 
    \end{example}

The following proposition generalises the above example to all well-behaved queries:

\begin{proposition}
\label{prop:pol_vo_dynamic_q-hierarchical}
    Any well-behaved conjunctive query has a well-structured VO.
\end{proposition}

Next, we define the {\em preprocessing width} $\fw$ of a VO $\omega$ and a query $Q$ in $\pol$. 
For a variable $X$ in a VO $\omega$,  let $Q_X$ denote the query, whose body is the conjunction of all atoms in 
$\omega_X$ and whose free variables are all its variables.
The preprocessing width of $\omega$ is defined next using the fractional edge cover number\footnote{The fractional edge cover number is used to define an upper bound on both the size of the query result and the time to compute it (Section~\ref{sec:prelims}).
} for such queries $Q_X$, whereas the preprocessing width of a query is the minimum over the preprocessing widths of its well-structured VOs:
\begin{align*}
\fw(\omega)  = \max_{X \in \vars(\omega)} 
\rho_{Q_X}^*(\{X\} \cup \dep_{\omega}(X))  \hspace{2em} \text{and} \hspace{2em} \fw(Q) = \min_{\omega \in \WVO(Q)} \ \fw(\omega)
\end{align*}

\begin{example}
\rm
Figure \ref{fig:pol_class_construction} depicts a well-structured 
VO $\omega$ for the query $Q$. 
We have $dep_{\omega}(D) = \{A,C\}$ and
$\rho_{Q_D}^*(\{A,C,D\}) = 2$. Overall, 
the preprocessing width of $Q$  is $2$.
\end{example}

The preprocessing width of any query in $\lin$ is $1$:

\begin{proposition}
\label{prop:acyclic_preprocessing_width}
For any conjunctive query $Q$ in $\lin$, it holds $\fw(Q)= 1$.
\end{proposition}


\subsection{Safe Rewriting of Queries in $\pol$}
\label{sec:safe_rewriting_lin_pol}

We show that every query in $\pol$ has a safe rewriting 
using views. The time to compute the views is $\bigO{N^{\fw}}$,
where $N$ is the database size and $\fw$ is the
preprocessing width of the query. 

Prior work uses view trees to maintain queries in the all-dynamic setting~\cite{ICDT23_access_pattern, FIVM:VLDBJ:2023}. We adapt the view tree construction to the 
setting over both static and dynamic relations and show that for
queries in $\pol$, 
the resulting view trees are safe rewritings.

\begin{figure}[t]
	\centering
	\setlength{\tabcolsep}{3pt}
	\renewcommand{\arraystretch}{1.05}
	\renewcommand{\linenumber}{\makebox[2ex][r]{\rownumber\TAB}}
	\setcounter{magicrownumbers}{0}
	\begin{tabular}[t]{@{}c@{}c@{}l@{}}
		\toprule
		\multicolumn{3}{l}{$\rewrite$(\text{VO} $\nu$) : rewriting using views}   \\
		\midrule
		\multicolumn{3}{l}{\MATCH $\nu$:} \\
		\midrule
		\phantom{ab} & $\{\nu_i\}_{i \in [n]}$\hspace*{2.5em} & \linenumber 
        \RETURN $\{\,\rewrite(\nu_i)\,\}_{i \in [n]}$ \\[2pt]
 		\cmidrule{2-3} \\[-12pt]
  \phantom{ab} & $R(\bm Y)$\hspace*{2.5em} & \linenumber \RETURN $R(\bm Y)$ \\
		\cmidrule{2-3} \\
		             &
		\begin{minipage}[t]{0.15\linewidth}
			\vspace{-1em}
			\hspace*{-0.35cm}
			\begin{tikzpicture}[xscale=0.5, yscale=1]
				\node at (0,-2)  (n4) {$X$};
				\node at (-1,-3)  (n1) {$\nu_1$} edge[-] (n4);
				\node at (0,-3)  (n2) {$\ldots$};
				\node at (1,-3)  (n3) {$\nu_k$} edge[-] (n4);
			\end{tikzpicture}
		\end{minipage}
		             &
		\begin{minipage}[t]{0.8\linewidth}
			\vspace{-0.6cm}
	  \linenumber \IF $k = 1$ and $\omega_1$ is atom $R(\bm Y)$ \SPACE\RETURN $R(\bm Y)$ \\[1ex]    \linenumber \LET $T_i = \rewrite(\nu_i)$ with root view $V_i(\bm S_i)$,  \ $\forall i\in[k]$ \\[1ex]
        \linenumber \LET $\bm S = \{X\} \cup \dep_\omega(X)$ \\[1ex]
        \linenumber \LET $V_X(\bm S) =$ join of $V_1(\bm S_1), \ldots, V_k(\bm S_k)$ \\[1ex]
        \linenumber \LET  $V'_X(\bm S\setminus \{X\}) = V_X(\bm S)$ \\[1ex]
        \linenumber \RETURN $
                    \left\{
                    \begin{array}{c}    
                        \tikz {
                            \node at (1.2,1)  (n4) {$V_X(\bm S)$};
						\node at (0.6,0.3)  (n1) {$T_1$} edge[-] (n4);
						\node at (1.2,0.3)  (n2) {$\ldots$};
						\node at (1.8,0.3)  (n3) {$T_k$} edge[-] (n4);
                            \node[anchor=west] at (3, 0.65) {\text{ if $X$ has no parent in $\omega$}};
                            
						\node at (1.2, -0.5)  (m1) {$V_X'(\bm S \setminus \{X\})$};
						\node at (1.2,-1.25)  (m4) {$V_X(\bm S)$} edge[-] (m1);
						\node at (0.6,-1.95)  (m1) {$T_1$} edge[-] (m4);
						\node at (1.2,-1.95)  (m2) {$\ldots$};
						\node at (1.8,-1.95)  (m3) {$T_k$} edge[-] (m4);
                            \node[anchor=west] at (3, -1.25) {otherwise};
					}
                    \end{array}  \right. $\\[0.5ex]
 
		\end{minipage}                                              \\[2.75ex]
		\bottomrule
	\end{tabular}
	\caption{Rewriting a query using views following its VO $\omega$ by calling 
	$\rewrite$($\omega$).}
	\label{fig:view_tree_construction}
\end{figure}

Given a VO $\omega$ for a query $Q$ in $\pol$, 
the procedure $\rewrite$ in Figure~\ref{fig:view_tree_construction} rewrites $Q$ using views 
following a top-down traversal of $\omega$. Initially, the parameter VO $\nu$ is $\omega$.
If $\nu$ is a set of trees, the procedure creates a view tree for each tree in $\nu$ (Line 1). 
If $\nu$ is a single atom, the procedure returns the atom (Line 2).
If $\nu$ is a single tree with root $X$, it proceeds as follows.
If $X$ has only one child and this child is an atom, then
the procedure returns the atom (Line 3).
Otherwise, $X$ has several child views. In this case, the procedure creates
a join view $V_X$ with free variables $\{X\} \cup \dep_{\omega}(X)$. 
If $X$ has a parent node, the procedure also 
 adds on top of the join view $V_X$ a projection view $V_X'$ that projects away $X$ (Line 8).

\begin{example}
\rm
    Figure~\ref{fig:pol_class_construction} shows a well-structured VO $\omega$ for the query $Q$ 
    and the view tree 
    constructed from $\omega$ using the procedure \rewrite.
    Observe that we obtain 
    the view tree from $\omega$ by replacing each variable $X$ either 
    by a single projection view $V_X'(dep_{\omega}(X))$ or by
    a join view  $V_X(\{X\} \cup dep_{\omega}(X))$ and a projection 
    view $V_X'(dep_{\omega}(X))$ on top. 
    
    We can verify that the view tree satisfies all four properties of safe rewritings from 
    Definition~\ref{def:safe_rewriting}.  
    The set of free variables of each dynamic view includes the set of free variables of each of its siblings; for instance,  
    the sibling views $V'_{C}(A)$ and $V'_B(A)$ have the same set of free variables. 
    The views $V_A(A)$ and $R^d(A,B)$ encode the query result. 

    The time to compute the view $V_D(A,C,D)$ is quadratic in the database size.
    All other views only need linear time to 
    compute semi-joins or project away a variable.
    Thus, the materialisation time for this rewriting is quadratic.
By Proposition~\ref{prop:safe_rewriting_complexity}, this query thus admits constant update time and constant enumeration delay after quadratic preprocessing time.  
\end{example}

Given a well-structured VO $\omega$ for a query in $\pol$, the procedure $\rewrite$ outputs a safe rewriting for $Q$ whose materialisation time is a function of the preprocessing width of $\omega$.
\begin{proposition}
    \label{prop:safe_view_trees}
    For any conjunctive query $Q$ in $\pol$, VO $\omega$ in $\WVO(Q)$, and database of size $N$, 
    $\rewrite(\omega)$ is a safe rewriting for $Q$
    with $\bigO{N^{\fw}}$ materialisation time, where $N$ is the database size and $\fw$ is the preprocessing width of $\omega$. 
\end{proposition}

Using a VO $\omega$ from $\WVO(Q)$ with $\fw(\omega) = \fw(Q)$, Proposition~\ref{prop:safe_view_trees}
 immediately implies   Theorem~\ref{thm:polynomial}.
Together with Proposition~\ref{prop:acyclic_preprocessing_width}, which states that 
 the queries in $\lin$ have preprocessing width 1, 
 Proposition~\ref{prop:safe_view_trees} also implies the complexity upper bound in 
Theorem~\ref{thm:dichotomy_linear}.

%% file: evaluation_expo.tex
\section{Evaluation of Queries in $\expo$}
\label{sec:evaluation_expo}

In this section, we introduce our evaluation strategy for queries in $\expo$, specifically targeting those in $\expo \setminus \pol$ for which the approach from previous sections is not applicable. 
We begin by demonstrating our strategy with a simple query from this class.

\begin{example}
\rm
Consider the query $Q_3(A,B) = R^d(A), S^s(A,B), T^d(B)$ from Figure \ref{fig:queries}. It does not admit a safe rewriting, thus the evaluation strategy described in Section~\ref{sec:evaluation_lin_pol} cannot achieve constant update and constant enumeration delay.
At a first glance, $Q_3$ does not seem tractable as a single-tuple update can lead to linearly many changes to the query result. 
For example, an insert $+a$ to $R^d$ may be paired with linearly many $B$-values in $S^s$.
However, the result of $Q_3$ is always a subset of the static relation $S^s$. 
With additional preprocessing time, we can precompute the effect of each update:
We can construct a transition system whose states are the possible subsets of $S^s$ and the updates to the dynamic relations may cause state transitions.
In this system, each state enables the query result to be enumerated with constant delay, and transitioning between states occurs in constant time.
\end{example}

Each state corresponds to a database instance and the materialised query result for that instance. Since the static relations are the same across all states, the states only record the content of the query result and of the dynamic relations without any dangling tuples.

\nop{-- an instance over the dynamic schema of a given database that includes only tuples relevant for the given query and that is maximal in that sense.}
For clarity in this section, we interpret a database as a finite set of facts and apply standard set operations to databases. The static and dynamic parts of a database $D$, denoted $D^s$ and $D^d$, consist of the facts from the static and dynamic relations in $D$, respectively, with $D = D^s \cup D^d$.
Before explaining how to construct a transition system, we first define maximum dynamic databases.

\begin{definition}\label{def:maximum_dyn_db}
    Let a conjunctive query $Q$ in $\expo \setminus \pol$ and a database $D = D^s \cup D^d$, where $D^d=\{R_1^d,\ldots, R_k^d\}$. For each dynamic body atom $R_i^d(\bm X)$ in $Q$, let $S_i^d$ be the result of the sub-query $Q_{R_i}$ of $Q$ whose body is the conjunction of the static body atoms of $Q$ and whose free variables are $\bm X$.
    The {\em maximum dynamic database} $D^d_{\max}$ for query $Q$ and database $D$ is $D^d_{\max}=\{S_1^d,\ldots,S_k^d\}$.
    \nop{
    Consider a conjunctive query $Q$ in $\expo \setminus \pol$ and a database $D = D^s \cup D^d$.
    For each dynamic atom $R^d(\bm X)$ in $Q$, let $Q_{R}$ denote the sub-query of $Q$ consisting only of its static body atoms, with $\bm X$ as free variables. 
    
    The maximum dynamic database $D^d_{\max}$ for query $Q$ and database $D$ is 
    the version of $D^d$ where each dynamic relation $R^d$ contains the result of the sub-query $Q_R$ evaluated on $D$.
    }
\end{definition}

By the definition of $\expo\setminus\pol$, each variable from a dynamic atom $R^d(\bm X)$ appears in at least one static atom of $Q_R$. The sub-query $Q_R$ identifies the complete set of $R^d$-tuples that might contribute to the result of $Q$; any other $R^d$-tuples have no impact on the result of $Q$. 

\begin{example}
\rm
    Consider the query $Q_3(A,B) = R^d(A), S^s(A,B), T^d(B)$ and the database $D$ from Figure~\ref{fig:transition_system_Q3} (left). The maximum dynamic database is obtained by evaluating the queries $Q_R(A) = S^s(A,B)$ and $Q_T(B) = S^s(A,B)$ on $D$. As a result, the maximum dynamic database for query $Q_3$ and database $D$ consists of $R^d = \{ a_1 \}$ and $T^d = \{ b_1, b_2 \}$.
\end{example}

We next establish the upper bound on the size of a maximum dynamic database. 

\begin{proposition}\label{prop:max_dynamic_database}
    For any conjunctive query $Q$ in $\expo \setminus \pol$ and database $D$ of size $N$,
    the maximum dynamic database $D^d_{\max}$ for query $Q$ and database $D$ has   $\bigO{N^{\rho^*(\stat(Q))}}$ size.
\end{proposition}

We next define the transition system for a query in $\expo \setminus \pol$ and a database $D$. 

\begin{definition}\label{def:transition_system}
    Consider a conjunctive query $Q$ in $\expo \setminus \pol$ and a database $D = D^s \cup D^d$.
    Let $D^d_{\max}$ denote the maximum dynamic database for query $Q$ and database $D$.

    A transition state for query $Q$ and database $D$ is a pair $(I, R)$, where $I \subseteq D^d_{\max}$ and $R$ is the result of $Q$ on $D^s \cup I$.
    The set $\calS$ of transition states for query $Q$ and database $D$ is:
    $\calS = \{\, (I, R) \mid I \subseteq D^d_{\max}, R = Q(D^s \cup I) \,\}$.

    A transition system for query $Q$ and database $D$ is a tuple $(\calS, s_{init}, \calU, \delta)$, where $\calS$ is the set of transition states for $Q$ and  $D$, $s_{init} \in \calS$ is the initial state, $\calU$ is the set of single-tuple updates to $D^d$, and $\delta: \calS \times \calU \rightarrow \calS$ is a transition function that determines the next state based on the current state and an update.
    The initial state $s_{init}$ pairs the database $D^d \cap D^d_{\max}$ and the result of $Q$ on $D^s \cup (D^d \cap D^d_{\max})$.
\end{definition}

We construct this transition system during the preprocessing step. For each single-tuple update to a dynamic relation, the transition system is used to move from the current state to another. If the update does not belong to the maximum dynamic database, it has no effect, and the system remains in the same state. The query result is then enumerated from the current state. Next, we demonstrate our evaluation strategy for the query $Q_3$.

\begin{example}\label{ex:transition_system_Q3}
\rm    
    Figure~\ref{fig:transition_system_Q3} shows the states and the transition system constructed for the query $Q_3$ and the database $D=\{R^d,T^d,S^s\}$. The maximum dynamic database $D^d_{\max}$ for $Q_3$ and $D$, consisting of $R^d = \{ a_1 \}$ and $T^d = \{ b_1, b_2 \}$, has 8 possible subsets. Each state corresponds to one of these subsets and the result of $Q_3$ evaluated on that subset and the static relation $S^s$.

    The initial state $s_{init}$, denoted by the blue arrow, corresponds to the database $D^d \cap D^d_{max}$, which includes $R^d = \{ a_1 \}$ and $T^d = \emptyset$, along with the associated empty query result.

    The state transitions correspond to single-tuple updates (inserts or deletes) to the dynamic relations $R^d$ and $T^d$.
    Only updates involving tuples from $D^d_{\max}$ are relevant to the result of $Q_3$, specifically updates involving     
    the value $a_1$ in $R^d$ and the values $b_1$ and $b_2$ in $S^d$. All other updates have no effect on the query result because the updated values do not appear in $S^s$. For simplicity, we show only the state transitions for relevant updates in Figure~\ref{fig:transition_system_Q3}.

    For the same query and a database of size $N$, the size of $D^d_{\max}$ is the total number of distinct $A$- and $B$-values in $S^s$, thus $\bigO{N}$. The query result at any state is a subset of $S^s$ and can be computed in $\bigO{N}$ time.
    The transition system captures all possible subsets of $D^d_{\max}$, thus it has $2^{|D^d_{\max}|}$ states in total. The number of transitions from any state is bounded by $2\cdot |D^d_{\max}|$ when considering only updates relevant to the query result. 
     
    When constructing the transition system in this case, we create a global index for $\bigO{2^N}$ states, with a lookup time of $\bigO{N}$. For each state, we also build a local index to support transitions to $\bigO{N}$ neighbouring states. The local index is constructed by performing a lookup in $\bigO{N}$ time in the global index for each state transition. With $\bigO{N}$ possible transitions, constructing the local index takes $\bigO{N^2}$ time. The local index is defined over $\bigO{N}$ states and provides constant-time lookups, see the computational model in Section~\ref{sec:prelims}.
    
    The preprocessing time is given by the number of states times the creation time per state, so $\bigO{2^{N} \cdot N^2}$. The update time is constant since for each single-tuple update to a dynamic relation, we find in constant time the corresponding transition and move to the target state or ignore the update if it does not match any transition. We enumerate the tuples in the current state trivially with constant delay, as the query result is already materialised.
\end{example}

\nop{
}

\begin{figure}[t]
    \begin{minipage}{0.22\textwidth}
    \centering
    \small
    \begin{tikzpicture}[xscale=1, yscale=0.75]
    
    \node at (0.5,1) {Input database};
    
    \node (tableNameR) at (0,0) {$R^d$};
    
    \node (tableR) [below=-0.2cm of tableNameR] {
    \begin{tabular}{|c|}
    \hline
    $A$ \\
    \hline
    $a_1$ \\
    \hline
    \end{tabular}
    };
    
    \node (tableNameS) at (1,0) {$T^d$};
    
    \node (tableS) [below=-0.2cm of tableNameS] {
    \begin{tabular}{|c|}
    \hline
    $B$ \\
    \hline
    $b_3$\\
    \hline
    \end{tabular}
    };
    
    \node (tableNameT) at (0.5, -2) {$S^s$};
    
    \node (tableT) [below=-0.2cm of tableNameT] {
    \begin{tabular}{|c|c|}
    \hline
    $A$ & $B$ \\
    \hline
    $a_1$ & $b_1$ \\
    \hline
    $a_1$ & $b_2$ \\
    \hline
    \end{tabular}
    };
    
    \end{tikzpicture}
    \end{minipage}
    \begin{minipage}{0.77\textwidth}
    \centering
    \begin{tikzpicture}[>=Stealth, node distance=1cm and 1cm, xscale=3.2, yscale=1.55]
    
    \node[inner sep=0cm] at (-0.75,0) (state00) {
    \scriptsize
    \setlength{\tabcolsep}{3pt} 
    \begin{minipage}{1.8em}
        \centering
        \begin{tabular}{|c}
        \hline\\[-0.2cm]
        $A$ \\[0.05cm]
        \hline\\[-0.15cm]
         {\color{white} $x_i$} \\[0.1cm] 
        \hline
        \end{tabular}
    \end{minipage}%
    \begin{minipage}{1.8em}
        \centering
        \begin{tabular}{|c|}
        \hline\\[-0.2cm]
        $B$ \\[0.05cm]
        \hline\\[-0.15cm]
         {\color{white} $x_i$} \\[0.1cm] 
        \hline
        \end{tabular}
    \end{minipage}%
    \begin{minipage}{2.1em}
        \centering
        \begin{tabular}{c|}
        \hline\\[-0.2cm]
        $Q_3$ \\[0.05cm]
        \hline\\[-0.15cm]
        {\color{white} $x_2$} \\[0.1cm]
        \hline
        \end{tabular}
    \end{minipage}
    };
    
    \node[inner sep=0cm] at (0,-1) (state0_2) {
    \scriptsize
    \setlength{\tabcolsep}{3pt}
    \begin{minipage}{1.8em}
        \centering
        \begin{tabular}{|c}
        \hline\\[-0.2cm]
        $A$ \\[0.05cm]
        \hline\\[-0.15cm]
         {\color{white} $x_i$} \\[0.1cm] 
        \hline
        \end{tabular}
    \end{minipage}%
    \begin{minipage}{1.8em}
        \centering
        \begin{tabular}{|c|}
        \hline\\[-0.2cm]
        $B$ \\[0.05cm]
        \hline\\[-0.15cm]
        $b_2$ \\[0.1cm]
        \hline
        \end{tabular}
    \end{minipage}%
    \begin{minipage}{2.1em}
        \centering
        \begin{tabular}{c|}
        \hline\\[-0.2cm]
        $Q_3$ \\[0.05cm]
        \hline\\[-0.15cm]
        {\color{white} $x_2$} \\[0.1cm]
        \hline
        \end{tabular}
    \end{minipage}
    };
    
    \node[inner sep=0cm] at (0,1) (state10) {
    \scriptsize
    \setlength{\tabcolsep}{3pt}
    \begin{minipage}{1.8em}
        \centering
        \begin{tabular}{|c}
        \hline\\[-0.2cm]
        $A$ \\[0.05cm]
        \hline\\[-0.15cm]
        $a_1$ \\[0.1cm]
        \hline
        \end{tabular}
    \end{minipage}%
    \begin{minipage}{1.8em}
        \centering
        \begin{tabular}{|c|}
        \hline\\[-0.2cm]
        $B$ \\[0.05cm]
        \hline\\[-0.15cm]
        {\color{white} $x_i$} \\[0.1cm] 
        \hline
        \end{tabular}
    \end{minipage}%
    \begin{minipage}{2.1em}
        \centering
        \begin{tabular}{c|}
        \hline\\[-0.2cm]
        $Q_3$ \\[0.05cm]
        \hline\\[-0.15cm]
        {\color{white} $x_2$} \\[0.1cm]
        \hline
        \end{tabular}
    \end{minipage}
    };
    
    \node[inner sep=0cm] at (0,0) (state01) {
    \scriptsize
    \setlength{\tabcolsep}{3pt}
    \begin{minipage}{1.8em}
        \centering
        \begin{tabular}{|c}
        \hline\\[-0.2cm]
        $A$ \\[0.05cm]
        \hline\\[-0.15cm]
        {\color{white} $x_i$} \\[0.1cm] 
        \hline
        \end{tabular}
    \end{minipage}%
    \begin{minipage}{1.8em}
        \centering
        \begin{tabular}{|c|}
        \hline\\[-0.2cm]
        $B$ \\[0.05cm]
        \hline\\[-0.15cm]
        $b_1$ \\[0.1cm]
        \hline
        \end{tabular}
    \end{minipage}%
    \begin{minipage}{2.1em}
        \centering
        \begin{tabular}{c|}
        \hline\\[-0.2cm]
        $Q_3$ \\[0.05cm]
        \hline\\[-0.15cm]
        {\color{white} $x_2$} \\[0.1cm]
        \hline
        \end{tabular}
    \end{minipage}
    };
    
    \node[inner sep=0cm] at (1,1) (state11) {
    \scriptsize
    \setlength{\tabcolsep}{3pt}
    \begin{minipage}{1.8em}
        \centering
        \begin{tabular}{|c}
        \hline\\[-0.2cm]
        $A$ \\[0.05cm]
        \hline\\[-0.15cm]
        $a_1$ \\[0.1cm]
        \hline
        \end{tabular}
    \end{minipage}%
    \begin{minipage}{1.8em}
        \centering
        \begin{tabular}{|c|}
        \hline\\[-0.2cm]
        $B$ \\[0.05cm]
        \hline\\[-0.15cm]
        $b_1$ \\[0.1cm]
        \hline
        \end{tabular}
    \end{minipage}%
    \begin{minipage}{3.2em}
        \centering
        \begin{tabular}{c|}
        \hline\\[-0.2cm]
        $Q_3$ \\[0.05cm]
        \hline\\[-0.15cm]
        $a_1, b_1$ \\[0.1cm]
        \hline
        \end{tabular}
    \end{minipage}
    };
    
    \node[inner sep=0cm] at (1,-1) (state0_12) {
    \scriptsize
    \setlength{\tabcolsep}{3pt}
    \begin{minipage}{1.8em}
        \centering
        \begin{tabular}{|c}
        \hline\\[-0.2cm]
        $A$ \\[0.05cm]
        \hline\\[-0.15cm]
        {\color{white} $x_i$} \\
        {\color{white} $x_i$} \\[0.1cm] 
        \hline
        \end{tabular}
    \end{minipage}%
    \begin{minipage}{1.8em}
        \centering
        \begin{tabular}{|c|}
        \hline\\[-0.2cm]
        $B$ \\[0.05cm]
        \hline\\[-0.15cm]
        $b_1$ \\
        $b_2$ \\[0.1cm]
        \hline
        \end{tabular}
    \end{minipage}%
    \begin{minipage}{3.2em}
        \centering
        \begin{tabular}{c|}
        \hline\\[-0.2cm]
        $Q_3$ \\[0.05cm]
        \hline\\[-0.15cm]
        {\color{white}$a_1, b_1$} \\
        {\color{white}$a_1, b_2$} \\[0.1cm]
        \hline
        \end{tabular}
    \end{minipage}
    };
    
    \node[inner sep=0cm] at (1,0) (state1_2) {
    \scriptsize
    \setlength{\tabcolsep}{3pt}
    \begin{minipage}{1.8em}
        \centering
        \begin{tabular}{|c}
        \hline\\[-0.2cm]
        $A$ \\[0.05cm]
        \hline\\[-0.15cm]
        $a_1$ \\[0.1cm]
        \hline
        \end{tabular}
    \end{minipage}%
    \begin{minipage}{1.8em}
        \centering
        \begin{tabular}{|c|}
        \hline\\[-0.2cm]
        $B$ \\[0.05cm]
        \hline\\[-0.15cm]
        $b_2$ \\[0.1cm]
        \hline
        \end{tabular}
    \end{minipage}%
    \begin{minipage}{3.2em}
        \centering
        \begin{tabular}{c|}
        \hline\\[-0.2cm]
        $Q_3$ \\[0.05cm]
        \hline\\[-0.15cm]
        $a_1, b_2$ \\[0.1cm]
        \hline
        \end{tabular}
    \end{minipage}
    };

    \node[inner sep=0cm] at (1.85,0) (state1_12) {
    \scriptsize
    \setlength{\tabcolsep}{3pt} 
    \begin{minipage}{1.8em}
        \centering
        \begin{tabular}{|c}
        \hline\\[-0.2cm]
        $A$ \\[0.05cm]
        \hline\\[-0.15cm]
        $a_1$ \\
        \\[0.1cm]
        \hline
        \end{tabular}
    \end{minipage}%
    \begin{minipage}{1.8em}
        \centering
        \begin{tabular}{|c|}
        \hline\\[-0.2cm]
        $B$ \\[0.05cm]
        \hline\\[-0.15cm]
        $b_1$ \\
        $b_2$ \\[0.1cm]
        \hline
        \end{tabular}
    \end{minipage}%
    \begin{minipage}{3.2em}
        \centering
        \begin{tabular}{c|}
        \hline\\[-0.2cm]
        $Q_3$ \\[0.05cm]
        \hline\\[-0.15cm]
        $a_1, b_1$ \\
        $a_1, b_2$ \\[0.1cm]
        \hline
        \end{tabular}
    \end{minipage}
    };

    
    \draw[->, thick, blue] (-0.5,1.1) -- ([yshift=3pt]state10.west);
    
    \draw[->] ([xshift=2pt]state00.north) -- ([yshift=-2pt, xshift=1pt]state10.west) node[midway, below] {\scriptsize $+a_1$};
    \draw[<-] ([xshift=-1pt]state00.north) -- ([yshift=2pt, xshift=1pt]state10.west) node[midway, above] {\scriptsize $-a_1$};
    
    \draw[->] ([xshift=2pt]state00.south) -- ([yshift=2pt, xshift=1pt]state0_2.west) node[midway, above] {\scriptsize $+b_2$};
    \draw[<-] ([xshift=-1pt]state00.south) -- ([yshift=-2pt, xshift=1pt]state0_2.west) node[midway, below] {\scriptsize $-b_2$};
    
    \draw[->] ([yshift=2pt]state00.east) -- ([yshift=2pt, xshift=1pt]state01.west) node[midway, above] {\scriptsize $+b_1$};
    \draw[<-] ([yshift=-2pt]state00.east) -- ([yshift=-2pt, xshift=1pt]state01.west) node[midway, below] {\scriptsize $-b_1$};
    
    \draw[->, bend left] ([yshift=2pt]state10.east) to node[pos=0.4, below] {\scriptsize $+b_1$} ([yshift=2pt, xshift=1pt]state11.west);
    \draw[<-, bend left] ([yshift=5pt]state10.east) to node[pos=0.4, above] {\scriptsize $-b_1$} ([yshift=5pt, xshift=1pt]state11.west);
    
    \draw[->] ([xshift=8pt]state10.south) -- ([yshift=8pt, xshift=1pt]state1_2.west) node[pos=0.8, above] {\scriptsize $+b_2$};
    \draw[<-] ([xshift=4pt]state10.south) -- ([yshift=4pt, xshift=1pt]state1_2.west) node[pos=0.7, below] {\scriptsize $-b_2$};
    
    \draw[->] ([xshift=2pt]state01.north) -- ([yshift=-2pt, xshift=1pt]state11.west) node[pos=0.3, below] {\scriptsize $+a_1$};
    \draw[<-] ([xshift=-2pt]state01.north) -- ([yshift=2pt, xshift=1pt]state11.west) node[pos=0.1, above] {\scriptsize $-a_1$};
    
    \draw[->] ([xshift=2pt]state01.south) -- ([yshift=2pt, xshift=1pt]state0_12.west)   node[pos=0.3, above] {\scriptsize $+b_2$};
    \draw[<-] ([xshift=-2pt]state01.south) -- ([yshift=-2pt, xshift=1pt]state0_12.west) node[pos=0.15, below] {\scriptsize $-b_2$};
    
    \draw[->] ([xshift=8pt]state0_2.north) -- ([yshift=-8pt, xshift=1pt]state1_2.west) node[pos=0.85, below] {\scriptsize $+a_1$};
    \draw[<-] ([xshift=4pt]state0_2.north) -- ([yshift=-4pt, xshift=1pt]state1_2.west) node[pos=0.65, above] {\scriptsize $-a_1$};
    
    \draw[->, bend right] ([yshift=-2pt]state0_2.east) to node[pos=0.4, above] {\scriptsize $+b_1$} ([yshift=-2pt, xshift=1pt]state0_12.west);
    \draw[<-, bend right] ([yshift=-5pt]state0_2.east) to node[pos=0.4, below] {\scriptsize $-b_1$} ([yshift=-5pt, xshift=1pt]state0_12.west);
    
    \draw[->] ([yshift=2pt]state11.east) -- ([xshift=1pt]state1_12.north) node[midway, above] {\scriptsize $+b_2$};
    \draw[<-] ([yshift=-2pt]state11.east) -- ([xshift=-2pt]state1_12.north) node[midway, below] {\scriptsize $-b_2$};
    
    \draw[->] ([yshift=2pt]state1_2.east) -- ([yshift=2pt, xshift=1pt]state1_12.west) node[midway, above] {\scriptsize $+b_1$};
    \draw[<-] ([yshift=-2pt]state1_2.east) -- ([yshift=-2pt, xshift=1pt]state1_12.west) node[midway, below] {\scriptsize $-b_1$};
    
    \draw[->] ([yshift=2pt]state0_12.east) -- ([xshift=-2pt]state1_12.south) node[midway, above] {\scriptsize $+a_1$};
    \draw[<-] ([yshift=-2pt]state0_12.east) -- ([xshift=1pt]state1_12.south) node[midway, below] {\scriptsize $-a_1$};
    
    \end{tikzpicture}
    \end{minipage}
    \caption{
    Left: Input database with dynamic relations $R^d(A)$ and $T^d(B)$ and static relation $S^s(A,B)$.
    Right: The transition system built from the input database and the query $Q_3(A,B) = R^d(A), S^s(A,B), T^d(B)$.
    Each state (box) consists of a dynamic database (first two columns), showing the possible contents of $R^d$ and $T^d$, restricted to the $A$- and $B$-values from $S^s$, respectively; and the materialised result of $Q_3$ on that dynamic database and $S^s$ (third column). The initial state  is derived from the input database (blue arrow).
    Each transition (arrow) denotes an insertion $(+a)$ or deletion $(-a)$ of an $A$-value $a$ in $R^d$, or an insertion $(+b)$ or deletion $(-b)$ of a $B$-value $b$ in $T^d$. 
    }
    \label{fig:transition_system_Q3}
\end{figure}

Consider now an arbitrary query from $\expo \setminus \pol$ and a database $D$ of size $N$.
The constructed transition system consists of $2^{p}$ states, where $p = |D^d_{\max}| = \bigO{N^{\rho^*(\stat(Q))}}$, per Definition~\ref{def:transition_system} and Proposition~\ref{prop:max_dynamic_database}.
For each state, the corresponding query result can be computed in $\bigO{N^{\rho^*(\stat(Q))}}$ time using a worst-case optimal join algorithm~\cite{Ngo:JACM:18}. 

We next construct the transitions for each state.
Each state has at most $2p$ transitions to other states when considering only updates relevant to the query result. 
The global index over the exponentially many states takes $2^p$ time to construct and requires time proportional to  $p$ for a lookup. 
For each state, it takes time quadratic in $p$ to construct a local index that comprises at most $2p$ neighbouring states using the global index with lookup time proportional to $p$. Overall, constructing the transition system takes time $2^p \cdot p^2$, where $p = \bigO{N^{\rho^*(\stat(Q))}}$. This matches the complexity in Theorem~\ref{thm:dichotomy_possible}.    

%% file: lower_bounds.tex
\section{Lower Bound for Queries Outside $\lin$}
\label{sec:lower_bounds}
In this section, we outline the proof of the lower bound in 
Theorem~\ref{thm:dichotomy_linear}.
The proof consists of two parts.
In the first part, we give a lower bound on the  complexity of evaluating the simple
queries $Q_{RST}() = R^d(A),S^s(A,B), T^d(B)$ 
and
$Q_{ST}(A) = S^s(A,B),T^d(B)$.
None of these queries is well-behaved and, therefore, not contained in $\lin$:
the first one has the path $(A,B)$ that connects 
the dynamic body atoms $R^d(A)$ and $T^d(B)$ but is not body-safe; 
the second one has the path $(B,A)$ that connects 
the dynamic body atom $T^d(B)$ with the head atom $Q_{ST}(A)$ but is not head-safe. 
In the second part of the proof, we show a lower bound 
on the  complexity of evaluating arbitrary conjunctive queries that are not contained in $\lin$ and 
do not have repeating
relation symbols.
The argument is as follows.
Consider a conjunctive query $Q \notin \lin$ that does not have repeating relation symbols. 
By the definition of $\lin$, this means that 
(1) $Q$ is not free-connex $\alpha$-acyclic, or 
(2) it has a path connecting two dynamic body atoms that is not body-safe, or 
(3) it has a path connecting a dynamic body atom with the head atom that is not head-safe.
In Case (1), we cannot achieve
constant-delay enumeration of the result after linear preprocessing time
(even without processing any update), unless the Boolean Matrix 
Multiplication conjecture fails~\cite{BaganDG07}. 
In Case (2), we reduce the evaluation of $Q_{RST}$ 
to the evaluation of $Q$.
In Case (3), 
we reduce the evaluation of $Q_{ST}$ to the evaluation of $Q$.
The latter two reductions transfer the lower bound 
for $Q_{RST}$ and $Q_{ST}$  to $Q$.
\nop{and are standard (see, e.g., \cite{BerkholzKS17}).}
In the following, we outline the lower bound proofs for $Q_{RST}$ and 
$Q_{ST}$ and defer further details to Appendix~\ref{app:lower_bounds}.

The lower bound for $Q_{RST}$ is conditioned on the 
Online Vector-Matrix-Vector Multiplication (OuMv) conjecture, which is 
implied by the Online Matrix-Vector Multiplication (OMv) conjecture~\cite{Henzinger:OMv:2015} 
(Section~\ref{sec:prelims}).
The proof of the lower bound is inspired by prior work, which
reduces the OuMv problem to the evaluation of the variant 
of $Q_{RST}$ where
all relations are dynamic~\cite{BerkholzKS17}.
We explain the differences of our reduction to prior work in terms of 
construction and implication.  
The reduction in prior work starts with an empty database 
and encodes the matrix 
of the OuMv problem into the relation $S$ using updates. 
In our case, it is not possible to do this encoding using updates, since the relation is static. Instead, we do the encoding
before the preprocessing stage of
the evaluation algorithm for $Q_{RST}$.
Since the preprocessing procedure of the evaluation algorithm for $Q_{RST}$
is executed on a non-empty database, we need to put a bound on the preprocessing time. 
Hence, while prior work establishes  a lower bound on the update time and enumeration delay 
 regardless of the preprocessing time, our reduction 
 implies a lower bound  on the combination of
 preprocessing time, update time, and 
enumeration delay:

\begin{proposition}
\label{prop:Q_RST_hard}
There is no algorithm that evaluates the conjunctive query 
$Q_{RST}()$ $=$ $R^d(A),$ $S^s(A,B),$ $T^d(B)$ 
with 
$\bigO{N^{3/2-\gamma}}$ preprocessing time, 
$\bigO{N^{1/2-\gamma}}$ update time, and 
$\bigO{N^{1/2-\gamma}}$ enumeration delay
for any $\gamma >0$, where $N$ is the database size, unless the OuMv conjecture fails.
\end{proposition}

\begin{proof}[Proof Sketch]
\nop{Prior work reduces the OuMv problem to the evaluation of the variant 
of $Q_{RST}$ where
all relations are dynamic~\cite{BerkholzKS17}.
The reduction starts with an empty database 
and encodes the matrix 
of the OuMv problem into the relation $S$ using updates. 
In our case, it is not possible to do this encoding using updates, since the relation is static. Instead, we do the encoding
before the preprocessing stage of
the evaluation algorithm for $Q_{RST}$.

We explain the reduction in more detail. 
}
Assume that there is an algorithm $\calA$ that 
evaluates $Q_{RST}$ with 
$\bigO{N^{1/2-\gamma}}$ update time and 
enumeration delay after $\bigO{N^{3/2-\gamma}}$ preprocessing time,  
for some $\gamma >0$. 
Consider an $n\times n$ matrix $M$ and $n$ pairs $(u_r,v_r)$ of 
$n$-dimensional vectors that serve as input to the OuMv problem.
We first encode $M$ into the relation $S$ in time $\bigO{n^2}$,
which leads to a database of size $\bigO{n^2}$.
Then, in each round $r \in [n]$, we  encode the vectors $u_r$
and $v_r$ into $R$ and respectively  $T$ using $\bigO{n}$ updates and trigger the enumeration procedure of 
$\calA$ to obtain from $Q_{RST}$ the result of $u_rMv_r$. 
This takes $\bigO{n(n^2)^{1/2-\gamma}}
= \bigO{n^{2-2\gamma}}$ time. After 
$n$ rounds, we use overall 
$\bigO{n^{3-2\gamma}}$ time. This means
that we solve the OuMv problem in sub-cubic time,
which contradicts the OuMv conjecture. 
\end{proof}

The reduction of the OMv problem to the evaluation of the query 
$Q_{ST}(A)$ $=$ $S^s(A, B),$ $T^d(B)$ 
is similar to the above reduction. 
We encode the matrix $M$ into the relation $S$ before the preprocessing stage  and encode each incoming vector $v_r$ into $T$ using updates. 

%% file: beyond_exp.tex
\section{Outlook: Tractability Beyond $\expo$}
\label{sec:outlook}

This work explores the tractability of conjunctive queries  over static and dynamic relations. 
The largest class of tractable queries put forward is $\expo$. Yet a characterisation of {\em all} tractable 
queries remains open.

In the following, we discuss the evaluation of queries outside $\expo$. 
Let a conjunctive query $Q$. The {\em reduced dynamic sub-query} of $Q$ is obtained from $Q$ by omitting all static atoms and all of their variables.\footnote{In contrast, the dynamic sub-query of $Q$ as defined in Section~\ref{sec:prelims} retains all variables that appear in at least one dynamic atom.}   
An immediate observation is that queries whose reduced dynamic sub-query is not $q$-hierarchical are not tractable. This is implied by the intractability of non-$q$-hierarchical queries in the all-dynamic setting~\cite{BerkholzKS17}. 
One example query whose reduced sub-query is not $q$-hierarchical is
$Q_6(A,B)$ $=$ $R^d(A),$ $S^d(A,B),$ $T^d(B,C),$ $U^s(C)$ from Figure~\ref{fig:queries}.
Its reduced dynamic sub-query is $Q_6'(A,B)$ $=$ $R^d(A),$ $S^d(A,B),$ $T^d(B)$. 
Any dynamic evaluation strategy for $Q_6$, where the variable $C$ is fixed to an arbitrary constant, translates into a dynamic evaluation strategy for $Q_6'$.   
This means that any complexity lower bound for the dynamic evaluation of $Q_6'$
is also a lower bound for $Q_6$.

A question is whether all queries, whose reduced dynamic sub-queries are $q$-hierarchical, are tractable. 
We discuss two such queries from Figure~\ref{fig:queries} that are not in $\expo$: 
$Q_4(A,B,C)$ $=$ $R^d(A,B),$ $S^d(A,C),$ $T^s(B,C)$ and 
$Q_5(B,C)$ $=$ $R^d(A,B),$ $S^d(A,C),$ $T^s(B,C)$.
The evaluation strategy from Example~\ref{ex:transition_system_Q3} can be easily extended to $Q_4$. 
Similar to Example~\ref{ex:transition_system_Q3}, we create a transition system where each state stores 
$(B,C)$-values and assign each $A$-value $a$ that is common to $R^d$ and $S^d$ to the state that stores the $(B,C)$-values paired with $a$ in the result.
Any update to $R^d$ or $S^d$ changes the assignment of at most one $A$-value. 
At any time, we can enumerate the query result by iterating over the $A$-values and enumerating for each of them the tuples in their corresponding state. So $Q_4$ is tractable, albeit not in $\expo$.

The query $Q_5(B,C)$ differs from $Q_4$ in that the variable $A$ is bound. The above approach for $Q_4$ does not allow for constant-delay enumeration of the result of $Q_5$ since distinct $A$-values might be assigned to distinct states that share $(B,C)$-values. Filtering out duplicates can however incur a non-constant enumeration delay. 
Therefore, our approach for $\expo$ cannot evaluate $Q_5$ tractably.

%% file: app_intro.tex
\section{Missing Details in Section~\ref{sec:intro}}
\label{app:intro}
\subsection{Proof of Theorem \ref{thm:dichotomy_linear}}
{\bf Theorem~\ref{thm:dichotomy_linear}.}
{\it 
Let a conjunctive query $Q$ and a database of size $N$. 
\begin{itemize}
\item If $Q$ is in $\lin$, then it can be evaluated with $\bigO{N}$ preprocessing time,
$\bigO{1}$ update time, and $\bigO{1}$ enumeration delay. 
\item If $Q$ is not in $\lin$ and does not have repeating relation symbols, then 
it cannot be maintained with  $\bigO{N}$ preprocessing time, $\bigO{1}$
update time, and $\bigO{1}$ enumeration delay, 
unless the Online Matrix-Vector Multiplication or the Boolean Matrix-Matrix Multiplication conjecture fail.  
\end{itemize}
}

\medskip
Consider a conjunctive query $Q$ and a database of size $N$.
 
We prove the first statement of Theorem~\ref{thm:dichotomy_linear}.  
Assume that $Q$ is in $\lin$. 
Let $\omega$ be a well-structured VO for $Q$
with minimal preprocessing width. From Proposition~\ref{prop:acyclic_preprocessing_width}, 
it follows that $\fw(\omega) =1$. Proposition~\ref{prop:safe_view_trees} implies that
 $Q$ has a safe rewriting  with $\bigO{N}$ materialisation time.
Using Proposition~\ref{prop:safe_rewriting_complexity}, we conclude that
$Q$ can be evaluated with $\bigO{N}$ preprocessing time, $\bigO{1}$ update time, and $\bigO{1}$ enumeration delay.  
     
The second statement of Theorem~\ref{thm:dichotomy_linear} follows from Proposition~\ref{prop:lower_bound_dichotomy} in Appendix~\ref{app:lower_bounds}.

\subsection{Proof of Theorem \ref{thm:polynomial}}
{\bf Theorem~\ref{thm:polynomial}.}
{\it 
    Let a conjunctive query $Q \in \pol$,  the preprocessing width $\fw$ of $Q$, and a database of size $N$. 
    The query $Q$  can be maintained with $\bigO{N^{\fw}}$ preprocessing time, $\bigO{1}$ update time, and $\bigO{1}$ enumeration delay. 
}

\medskip

Consider a conjunctive query $Q \in \pol$ with preprocessing width $\fw$  and a database of size $N$.
Let $\omega$ be a well-structured VO for $Q$
with minimal preprocessing width. It holds $\fw = \fw(\omega)$. 
It follows from Proposition~\ref{prop:safe_view_trees} that
 $Q$ has a safe rewriting  with $\bigO{N^{\fw}}$ materialisation time.
Proposition~\ref{prop:safe_rewriting_complexity}, implies that 
$Q$ can be evaluated with $\bigO{N^{\fw}}$ preprocessing time, $\bigO{1}$ update time, and $\bigO{1}$ enumeration delay.

\subsection{Proof of Theorem \ref{thm:dichotomy_possible}}
{\bf Theorem~\ref{thm:dichotomy_possible}.}
{\it 
Let a conjunctive query $Q$, the static sub-query $\stat(Q)$ of $Q$,  the fractional edge cover number $\rho^*(Q)$ of $Q$, $p = \bigO{N^{\rho^*(\stat(Q))}}$, and a database whose static relations have size $N$. 
If $Q$ is in $\expo$, then it can be maintained with $2^p\cdot p^2$ preprocessing time, $\bigO{1}$ update time, and $\bigO{1}$ enumeration delay. 
}

\medskip

The theorem follows Proposition~\ref{prop:max_dynamic_database} and the explanations in Section~\ref{sec:evaluation_expo}.

%% file: app_rewriting.tex
\section{Proof of Proposition \ref{prop:safe_rewriting_complexity}}
{\bf Proposition~\ref{prop:safe_rewriting_complexity}}
{\it 
Let a conjunctive query $Q$ and a database of size $N$.
If $Q$ has a safe rewriting with $f(N)$ materialisation time for some function $f$, then $Q$ can be evaluated with $f(N)$ preprocessing time, $\bigO{1}$ update time, and $\bigO{1}$ enumeration delay. 
}

\medskip

Let $\calT$ be a safe rewriting for $Q$. 
The preprocessing time is the materialisation  time for $\calT$.

To enumerate the result of $Q$, we nest the enumeration procedures for the connected components of $Q$, concatenating their result tuples.
For each connected component $C$, the enumeration procedure traverses the tree $T \in \calT$ containing all atoms from $C$ in a top-down manner.
The enumeration property of $\calT$ guarantees that there exists a subtree $T'$ of $T$ such that 
the root of $T'$  is the same as the root of $T$ and the set of free variables of the views in 
$T'$ consists of the free variables of $C$.

We enumerate the result of $C$ by traversing the subtree $T'$ in preorder. 
At each view $V(\bm X)$, we fix the values of the variables in $\bm Y$, where $\bm Y$ is the set of free variables of the ancestor views of $V$. 
We retrieve in constant time a tuple of values over $\bm X \setminus \bm Y$ from $V$ for the given $\bm Y$-value. 
After visiting all views once, we construct the first complete result tuple for $C$ and report it. 
We continue iterating over the remaining distinct values over $\bm X \setminus \bm Y$ in the last visited view $V$, reporting new tuples with constant delay.
After exhausting all values from $V$, we backtrack and repeat the enumeration procedure for the next $\bm Y$-value. The enumeration stops once all views from the subtree $T'$ are exhausted.
Given that all views are calibrated bottom-up but the enumeration proceeds top-down, the procedure only visits those tuples that appear in the  result, thus ensuring constant enumeration delay.

We propagate a constant-sized update through a tree in a bottom-up manner, maintaining each view on the path from the affected relation to the root.
From the update property of the safe rewriting $\calT$, computing the delta of any join view involves only constant-time lookups in the sibling views of the child carrying the update. The size of the delta also remains constant.
Computing the delta of a projection view also requires a constant-time projection of its incoming update.
Since an update to one relation affects one tree of $\calT$, propagating an update through $\calT$ takes constant time. 

%% file: app_query_classes.tex
\section{Missing Details in Section \ref{sec:query_classes}}
\label{app:query_classes}

\subsection{Proof of Proposition \ref{prop:safety_properites_poly_time_check}}
{\bf Proposition~\ref{prop:safety_properites_poly_time_check}.}
{\it For any conjunctive query $Q$ with $n$ variables and $m$ atoms, we can decide in $O(n^2 \cdot m^2)$ time whether $Q$ is well-behaved.
}

\medskip
We first explain how we can check that every path connecting two dynamic body atoms in $Q$ is body-safe.
For each pair $(R(\bm X), S(\bm Y))$ of dynamic body atoms of $Q$,
we first construct the Gaifman graph of the hypergraph of $Q$ without the variables $\bm C = \bm X \cap \bm Y$. The graph contains $\bigO{n}$ nodes and $\bigO{n^2}$ edges. 
We next choose any $x \in \bm X \setminus \bm C$ and $y \in \bm Y \setminus \bm C$ and check
in $\bigO{n^2}$ time
 if $x$ is reachable from $y$ using the Breadth-First Search algorithm; 
 if so, $Q$ has a path 
 connecting $R(\bm X)$ and $S(\bm Y)$ that is not body-safe.
We repeat this procedure for every pair of dynamic body atoms, which gives the total cost of $\bigO{n^2 \cdot m^2}$.

We now show how to check that every path connecting a dynamic body atom with the head atom is head-safe.
For each dynamic body atom $R(\bm X)$, 
we construct the Gaifman graph of $Q$ without the free variables in $\bm X$.
For each atom containing a free variable $y$, 
we choose any $x \in \bm X \setminus \free(Q)$ and check if $x$ is reachable from $y$ in the graph; 
if so, $Q$ has path connecting $R(\bm X)$ with the head atom that is not head-safe.
\nop{if so, $Q$ has unsafe atom-to-variable paths.}  
The total time is $\bigO{n^2\cdot m^2}$.

\subsection{Proof of Proposition \ref{prop:safety_properties_q-hierarchical}}
{\bf Proposition~\ref{prop:safety_properties_q-hierarchical}.}
{\it 
\begin{itemize}
\item Any conjunctive query without static relations is q-hierarchical if and only if it
is well-behaved. 
\item The dynamic sub-query of any well-behaved conjunctive query is q-hierarchical. 
\end{itemize}
}

\medskip
We start with the proof of the second statement in Proposition \ref{prop:safety_properties_q-hierarchical}. 
Consider a well-behaved query $Q$. 
For the sake of contradiction, assume that $\dyn(Q)$ is not q-hierarchical.
This means that one of the following two cases holds:
(1) $\dyn(Q)$ is not hierarchical or 
(2) $\dyn(Q)$ is hierarchical but not q-hierarchical.
We show that both cases lead to a contradition.

Assume that $\dyn(Q)$ is not hierarchical. 
In the following, we denote by $\dynatoms(X)$
the set of dynamic body atoms in $Q$ that contain the variable $X$.
The query $Q$ must have
two variables $X$ and $Y$ such that 
$\dynatoms(X) \not\subseteq \dynatoms(Y)$, 
$\dynatoms(Y) \not\subseteq \dynatoms(X)$, and 
$\dynatoms(X) \cap \dynatoms(Y) \neq \emptyset$.
This implies that $Q$ has three dynamic body atoms 
$R^d(\bm X)$, $S^d(\bm Y)$, and $T^d(\bm Z)$
such that 
$X \in \bm X$, $X \in \bm Y$, $X \notin \bm Z$,
$Y \in \bm Z$, $Y \in \bm Y$, and $Y \notin \bm X$.
The path $\bm P = (X,Y)$ connects the two dynamic
atoms $R^d(\bm X)$ and $T^d(\bm Z)$ such that 
$\bm P \cap \bm X \cap \bm Z = \emptyset$.
This means that $\bm P$ is not body-safe, which is 
a contradiction to our assumption that $Q$ is well-behaved.

Assume now that $\dyn(Q)$ is hierarchical, but not $q$-hierarchical. 
This implies that $Q$ contains 
two dynamic body atoms $R^d(\bm X)$ and $S^d(\bm Y)$, a bound variable $X$
with $X \in \bm X$ and $X \in \bm Y$, and a free variable 
$Y$ with $Y \in \bm Y$ and $Y \notin \bm X$.
The path $\bm P = (X,Y)$ connects the dynamic
body atom $R^d(\bm X)$ with the head atom of $Q$ such that 
$\bm P \cap \free(Q) \cap \bm X = \emptyset$.
This means that the path $\bm P$ is not head-safe, which is 
again a contradiction.


Next, we prove the first statement in Proposition \ref{prop:safety_properties_q-hierarchical}. The "if"-direction 
follows directly from the second statement of the proposition. It remains to show the 
"only-if"-direction. We prove the contraposition of this direction. 
Consider a  query $Q$ without static relations 
that is not well-behaved. We show that $Q$ is not q-hierarchical. 

Assume that $Q$ has a path $\bm P = (X_1, \ldots , X_n)$
that connects two body atoms $R^d(\bm X)$ and $T^d(\bm Z)$ such that 
$\bm P \cap \bm X \cap \bm Y = \emptyset$, i.e., 
$\bm P$ is not body-safe.
Assume that $\bm P$ is the shortest such path. 
Since $\bm P$ is not body-safe, it must hold $n>1$.
Furthermore, every pair $(X_i, X_{i+1})$ with $i \in [n-1]$ is contained in a body atom
i.e., $\atoms(X_i) \cap \atoms(X_{i+1}) \neq \emptyset$.
The following claim implies that $Q$ is not hierarchical, hence, not q-hierarchical:

\smallskip
\noindent
Claim 1: {\it $\exists i \in [n-1]$ such that  $\atoms(X_{i}) \not\subseteq \atoms(X_{i+1})$ and 
$\atoms(X_{i+1}) \not\subseteq \atoms(X_{i})$.}
\smallskip

We prove Claim 1.  
Assume that for each $i \in [n-1]$, it holds $\atoms(X_{i}) \subseteq \atoms(X_{i+1})$ or
$\atoms(X_{i+1}) \subseteq \atoms(X_{i})$. Consider an $i \in [n-1]$ such that 
$\atoms(X_{i}) \subseteq \atoms(X_{i+1})$ (the case for $\atoms(X_{i+1}) \subseteq \atoms(X_{i})$
is treated analogously). If $i = 1$, let $\bm P' = (X_{i+1}, \ldots , X_{n})$.
Otherwise, let $\bm P'$ $=$ $(X_1,$ $\ldots ,$ $X_{i-1},$ $X_{i+1},$
$\ldots ,$ $X_{n})$.
The path $\bm P'$ is shorter than $\bm P$ and still not body-safe. 
This contradicts our assumption
on the length of $\bm P$.

\smallskip
Assume now that $Q$ has a path $\bm P = (X_1, \ldots , X_n)$
connecting a dynamic body atom $R^d(\bm X)$ with the head atom of $Q$ 
such that $\bm P \cap \bm X \cap \free(Q) = \emptyset$, i.e., $\bm P$
is not head-safe. Let $\bm P$ be the shortest such path. 
Without loss of generality, we assume that $Q$ is hierarchical.
Since $X_n$ is free,
the following claim implies that $Q$ is not q-hierarchical: 

\smallskip
\noindent
Claim 2: {\it $X_{n-1}$ is bound and $\atoms(X_{n-1}) \supset \atoms(X_{n})$.}
\smallskip

We prove Claim 2. 
 Assume that $X_{n-1}$ is free. Since $\bm P$ is not head-safe, $X_{n-1}$ cannot be included in $\bm X$.
 Hence, $\bm P' = (X_1, \ldots , X_{n-1})$ is a path that is shorter than $\bm P$ and not 
 head-safe, which contradicts our assumption on the length of $\bm P$.
Assume now that $\atoms(X_{n-1}) \not\supset \atoms(X_{n})$. Since $Q$ is hierarchical and there must be at least 
one body atom including both $X_{n-1}$ and $X_n$, it must hold $\atoms(X_{n-1}) \subseteq \atoms(X_{n})$.
This implies that $n >2$, since otherwise the free variable $X_n$ is included in $\bm X$, which contradicts our assumption that $\bm P$ is 
not head-safe.  
This means however that the path $\bm P'' = (X_1, \ldots , X_{n-2}, X_{n})$ is sorter than $\bm P$ and 
not head-safe, which contradicts our assumption on the length of $\bm P$.

%% file: app_evaluation_lin_pol.tex
\section{Missing Details in Section~\ref{sec:evaluation_lin_pol}}
\label{app:evaluation_lin_pol}
We introduce the notion of {\em static parts}  of conjunctive queries and VOs tailored to them. These concepts will be useful in the proofs of Propositions~\ref{prop:pol_vo_dynamic_q-hierarchical}  and
\ref{prop:acyclic_preprocessing_width}.

\nop{Without loss of generality, 
we consider in this section conjunctive queries without repeating relation symbols. In case a relation symbol $R$ appears $k >1$ times in a query, we can  create $k$ copies of $R$ in the database and refer to each copy by a distinct relation name.  

assume that $Q$ does not have repeating relation symbols. In case a relation symbol $R$
}

\paragraph{Static Parts of Conjunctive Queries}
Consider a conjunctive query $Q$ and a static body atom $A = S(\bm Y)$ in $Q$.
Let $\bm D$ be the set of variables in $Q$ that appear in dynamic body atoms.  
The {\em reduced variant} of $A$ is of the form 
$S_A(\bm Y')$, where $\bm Y' = \bm Y \setminus \bm D$.
Let $Q_r$ be the query that results from $Q$
by replacing each body atom by its reduced variant. 
Consider a connected component $\calC$ of body atoms in $Q_r$. 
Let $\calA = \{A| S_A(\bm Y') \in \calC\}$ be the set of all static body atoms in $Q$ whose reduced variants are contained in $\calC$. 
Let $\bm I$ be the set of variables that appear 
in $\calA$ and in $\dyn(Q)$,
i.e., 
$\bm I = (\bigcup_{S(\bm Y) \in \calA} \bm Y) \cap \bm D$.
Let $\bm F$ be the free variables of $Q$ that appear in $\calA$, 
i.e., 
$\bm F = (\bigcup_{S(\bm Y) \in \calA} \bm Y) \cap \free(Q)$.
We call $Q_p(\bm F) = (S(\bm Y))_{S(\bm Y) \in \calA}$
a {\em static part of $Q$ with intersection set $\bm I$}.

\begin{example}
\rm
Figure~\ref{fig:pol_class_construction_appendix} visualizes a well-behaved conjunctive query $Q$ (top left), its dynamic sub-query $\dyn(Q)$ (top middle), and the static parts of $Q$ together with their intersection sets given below (top right).
For instance, the third static part $Q_p(C) = R(G,C), S(G,F)$ in the figure has the intersection set $\{C,F\}$.
\end{example}


\begin{figure}[t]
    \centering
    \begin{tikzpicture}
    \node[draw=none,fill=none](a) at (0, 0) {$\underline{A}$};
    \node[draw=none,fill=none](b) at (-1, -1) {$\underline{B}$};
    \node[draw=none,fill=none](d) at (-2, -2) {$\underline{D}$};
    \node[draw=none,fill=none](c) at (1, -1) 
    {$\underline{C}$};
    \node[draw=none,fill=none](f) at (2, -2) {$F$};
    \node[draw=none,fill=none](f) at (2, 0) {$G$};    
    \node[draw=none,fill=none](e) at (0, -2) {$E$};

    \node[draw=none,fill=none](n) at (-2, -3) {$N$};
    \node[draw=none,fill=none](p) at (-2, -4) 
    {$\underline{P}$};

    \node[draw=none,fill=none](j) at (2, -3) {$J$};
    \node[draw=none,fill=none](l) at (1.5, -4) {$L$};
    \node[draw=none,fill=none](k) at (2.5, -4) {$K$};
    
    \draw [red] plot [smooth cycle] coordinates {(.4,0) (0,.4) (-1.2,-0.7) (-2.4,-2) (-2,-2.4) (-0.7,-1.3)};
    \draw [red] plot [smooth cycle] coordinates {(.3,0) (0,.3) (-1.3,-1) (-1,-1.3)};
    \draw [red] plot [smooth cycle] coordinates {(0,.3) (-.3,0) (.6,-1) (-.3,-2) (0,-2.3) (1.3, -1.3) (1.3,-0.7)};
    \draw [red] plot [smooth cycle] coordinates {(0,.4) (-.4,0) (.85,-1.4) (2,-2.4) (2.4,-2) (1.4,-0.65)};

    \draw [blue, dotted, line width = 1.2] plot [smooth cycle] coordinates {(-1, -.5) (-2.6, -1.8) (-2.3, -3.3) (-1.7, -3.3) (-1.3, -2.2) (-.5, -1)};
    \draw [blue, dotted, line width = 1.2] plot [smooth cycle] coordinates {(-2.2, -2.8) (-2.3, -4.1) (-2.2, -4.3) (-1.8, -4.3) (-1.7, -4.1) (-1.8, -2.8)};

    \draw[blue, dotted ,line width = 1.2] (2,-2.5) ellipse (0.3cm and 0.8cm);
        \draw[blue, dotted , line width = 1.2, rotate=25] (0.5,-4.2) ellipse (0.3cm and 0.9cm);
    \draw[blue, dotted , line width = 1.2, rotate=-28] (3.2,-2.3) ellipse (0.3cm and 0.9cm);

    \draw[blue, dotted , line width = 1.2, rotate=43] (0.7,-1.4) ellipse (1.1cm and 0.3cm);
    \draw[blue, dotted , line width = 1.2] (2,-1) ellipse (0.3cm and 1.3cm);
    
    \draw[blue, dotted , line width = 1.2, rotate=43] (-0.7,-1.4) ellipse (1.1cm and 0.3cm);

 \begin{scope}[xshift=5cm]
    
    \node[draw=none,fill=none](a) at (0, 0) 
    {$\underline{A}$};
    \node[draw=none,fill=none](b) at (-1, -1) 
    {$\underline{B}$};
    \node[draw=none,fill=none](d) at (-2, -2) 
    {$\underline{D}$};
    \node[draw=none,fill=none](c) at (1, -1) 
    {$\underline{C}$};
    \node[draw=none,fill=none](f) at (2, -2) {$F$};
    \node[draw=none,fill=none](e) at (0, -2) {$E$};
    
    \draw [red] plot [smooth cycle] coordinates {(.4,0) (0,.4) (-1.2,-0.7) (-2.4,-2) (-2,-2.4) (-0.7,-1.3)};
    \draw [red] plot [smooth cycle] coordinates {(.3,0) (0,.3) (-1.3,-1) (-1,-1.3)};
    \draw [red] plot [smooth cycle] coordinates {(0,.3) (-.3,0) (.6,-1) (-.3,-2) (0,-2.3) (1.3, -1.3) (1.3,-0.7)};
    \draw [red] plot [smooth cycle] coordinates {(0,.4) (-.4,0) (.85,-1.3) (2,-2.4) (2.4,-2) (1.4,-0.65)};
 \end{scope}

    \begin{scope}[xshift=10cm, yshift=0.5cm ]
    \node[draw=none,fill=none](b) at (-2, -1) {$\underline{B}$};
    \node[draw=none,fill=none](d) at (-2, -2) {$\underline{D}$};

    \node[draw=none,fill=none](n) at (-2, -3) {$N$};
    \node[draw=none,fill=none](p) at (-2, -4) 
    {$\underline{P}$};

   \draw[blue, dotted , line width = 1.2] (-2,-2) ellipse (0.3cm and 1.4cm);
   \draw[blue, dotted , line width = 1.2] (-2,-3.5) ellipse (0.3cm and 0.9cm);
    \end{scope}

\begin{scope}[xshift = 7.5cm, yshift=1.2cm]
    \node[draw=none,fill=none](f) at (2, -2) {F};
    \node[draw=none,fill=none](j) at (2, -3) {J};
    \node[draw=none,fill=none](l) at (1.5, -4) {L};
    \node[draw=none,fill=none](k) at (2.5, -4) {K};

    \draw[blue, dotted , line width = 1.2] (2,-2.5) ellipse (0.3cm and 0.8cm);
        \draw[blue, dotted , line width = 1.2, rotate=25] (0.5,-4.2) ellipse (0.3cm and 0.9cm);
    \draw[blue, dotted , line width = 1.2, rotate=-28] (3.2,-2.3) ellipse (0.3cm and 0.9cm);
\end{scope}
    \begin{scope}[xshift=6.5cm, yshift= 1cm]
    \node[draw=none,fill=none](c1) at (5, -2) {$G$};
    \node[draw=none,fill=none](f1) at (4.3, -3) {$\underline{C}$};
    \node[draw=none,fill=none](f1) at (5.7, -3) {$F$};    

    \draw[blue, dotted , line width = 1.2, rotate=-35] (5.2,0.6) ellipse (0.3cm and 1cm);
   \draw[blue, dotted , line width = 1.2, rotate=33] (3.1,-5) ellipse (0.3cm and 1cm);    
    \end{scope}
        \begin{scope}[xshift=8cm, yshift=1cm]
    \node[draw=none,fill=none](c1) at (5, -2) {$\underline{C}$};
    \node[draw=none,fill=none](f1) at (5, -3) {$E$};

\draw[blue, dotted , line width = 1.2] (5,-2.5) ellipse (0.3cm and 1cm);
    \end{scope}

  \begin{scope}[yshift=-6.5cm]
\node[draw=none,fill=none](a) at (0, 0) {$\underline{A}$};
    \node[draw=none,fill=none](b) at (-1, -1) {$\underline{B}$};
    \node[draw=none,fill=none](d) at (-2, -2) {$\underline{D}$};
    \node[draw=none,fill=none](c) at (1, -1) {$\underline{C}$};
    \node[draw=none,fill=none](f) at (2, -2) {$F$};
    \node[draw=none,fill=none](e) at (0, -2) {$E$};

    \node[draw=none,fill=none](p) at (-2, -3) {$\underline{P}$};
    \node[draw=none,fill=none](n) at (-2, -4) {$N$};
    
    \node[draw=none,fill=none](j) at (1, -3) {$J$};
    \node[draw=none,fill=none](l) at (0, -4) {$L$};
    \node[draw=none,fill=none](k) at (2, -4) {$K$};

    \node[draw=none,fill=none](g) at (2, -3) {$G$};

    \draw (a) -- (b);
    \draw (d) -- (b);
    \draw (a) -- (c);
    \draw (e) -- (c);
    \draw (f) -- (c);

    \draw (d) -- (p);
    \draw (p) -- (n);

    \draw (f) -- (j);
    \draw (l) -- (j);
    \draw (k) -- (j);

    \draw (f) -- (g);
  \end{scope}

    \begin{scope}[yshift=-7cm,xshift=4.5cm]
\node[draw=none,fill=none](a) at (0, 0) {$\underline{A}$};
    \node[draw=none,fill=none](b) at (-1, -1) {$\underline{B}$};
    \node[draw=none,fill=none](d) at (-2, -2) {$\underline{D}$};
    \node[draw=none,fill=none](c) at (1, -1) {$\underline{C}$};
    \node[draw=none,fill=none](f) at (2, -2) {$F$};
    \node[draw=none,fill=none](e) at (0, -2) {$E$};

    \draw (a) -- (b);
    \draw (d) -- (b);
    \draw (a) -- (c);
    \draw (e) -- (c);
    \draw (f) -- (c);
  \end{scope}

    \begin{scope}[yshift=-7.3cm, xshift=8cm]

    \node[draw=none,fill=none](b) at (0, 0) {$\underline{B}$};
    \node[draw=none,fill=none](d) at (0, -1) {$\underline{D}$};

    \node[draw=none,fill=none](p) at (0, -2) {$\underline{P}$};
    \node[draw=none,fill=none](n) at (0, -3) {$N$};

    \draw (b) -- (d);
    \draw (d) -- (p);
    \draw (p) -- (n);
  \end{scope}

      \begin{scope}[yshift=-7.3cm, xshift=9.5cm]

    \node[draw=none,fill=none](f) at (0, 0) {$F$};
    \node[draw=none,fill=none](j) at (0, -1) {$J$};

    \node[draw=none,fill=none](l) at (-0.5, -2) {$L$};
    \node[draw=none,fill=none](k) at (0.5, -2) {$K$};

    \draw (f) -- (j);
    \draw (j) -- (l);
    \draw (j) -- (k);
  \end{scope}

      \begin{scope}[yshift=-7.3cm, xshift=11.5cm]
    \node[draw=none,fill=none](c) at (0, 0) {$\underline{C}$};
    \node[draw=none,fill=none](f) at (0, -1) {$F$};
\node[draw=none,fill=none](g) at (0, -2) {$G$};

    \draw (c) -- (f);
    \draw (f) -- (g);    
  \end{scope}
        \begin{scope}[yshift=-7.3cm, xshift=13cm]

    \node[draw=none,fill=none](c) at (0, 0) {$\underline{C}$};
    \node[draw=none,fill=none](e) at (0, -1) {$E$};

    \draw (c) -- (e);
  \end{scope}

%
\node at (0,1) {$Q$};
\node at (5,1) {$\dyn(Q)$};
\node at (10.5,1) {static parts};
\draw [decorate,decoration={brace,amplitude=10pt},xshift=0pt,yshift=0pt]
    (7.9,0.2) -- (13.1,0.2) node [black,midway,yshift=-15pt] {};

    \node at (11,-5.1) {intersection sets};
\draw [decorate,decoration={mirror, brace,amplitude=10pt},xshift=0pt,yshift=0pt]
    (7.5,-4.5) -- (13.5,-4.5) node [black,midway,yshift=-15pt] {};

    \draw [decorate,decoration={brace,amplitude=10pt},xshift=0pt,yshift=0pt]
    (8,-6.9) -- (13.2,-6.9) node [black,midway,yshift=-15pt] {};

\node at (8, -4.2) {$\{B,D\}$};
\node at (9.35, -4.2) {$\{F\}$};
\node at (11.5, -4.2) {$\{C,F\}$};
\node at (13, -4.2) {$\{C,E\}$};

    \node at (0.1,-6) {$\omega$};
    \node at (4.6,-6) {$\omega_d$};
   \node at (10.4,-6) {$\omega_d$-compatible free-top VOs for static parts};
   \end{tikzpicture}
    \caption{Top row from left to right: Query $Q$, its dynamic sub-query $\dyn(Q)$, and its static parts with their intersection
    sets given below. Bottom row from left to right: Well-structured VO $\omega$ for $Q$, well-structured VO $\omega_d$ for $\dyn(Q)$, and
    $\omega_d$-compatible free-top VOs for the static parts. 
    For simplicity, the atoms in the VOs are omitted.}
    \label{fig:pol_class_construction_appendix}
\end{figure}


The intersection set of each  static part of a well-behaved conjunctive query is covered by a single dynamic body atom of the query:

\begin{lemma}
\label{lem:neck_dynamic-coverage}
For any well-behaved conjunctive query $Q$ 
and static part of 
$Q$ with intersection set $\bm I$, there is a dynamic
body atom $R^d(\bm X)$ in $Q$ such that $\bm I \subseteq \bm X$.   
\end{lemma}

\begin{proof}
Consider a well-behaved conjunctive query $Q$ and a static part $Q_p$
of $Q$ with intersection $\bm I$.
For the sake of contradiction, assume that we need at least two distinct dynamic body atoms $R_1^d(\bm X_1)$ and $R_2^d(\bm X_2)$ 
from $Q$ to cover the variables in 
$\bm I$. This means that there are variables $X_1 \in
(\bm X_1 \cap \bm I)$ and $X_2 \in (\bm X_2 \cap \bm I)$ such that $X_1 \notin \bm X_1 \cap \bm X_2$ and $X_2 \notin \bm X_1 \cap \bm X_2$. Since $Q_p$ arises from a connected component of $Q$ after removing variables that appear in dynamic body atoms,  there must be a path 
$\bm P = (X_1=Y_1, \ldots , Y_n = X_2)$ in $Q$, 
such that all variables $Y_i$ with $i \in \{2, \ldots, n-1\}$ appear only in static body atoms. Hence, none of the variables $Y_i$ with $i \in \{2, \ldots, n-1\}$ is included in the intersection 
$\bm X_1 \cap \bm X_2$. Since $X_1$ and $X_2$ are not included in $\bm X_1 \cap \bm X_2$ either, it holds 
$\bm P \cap \bm X_1 \cap \bm X_2 = \emptyset$. This means that $\bm P$ connects the two dynamic body atoms
$R_1^d(\bm X_1)$ and $R_2^d(\bm X_2)$ but is not body-safe.
Thus, $Q$ is not well-behaved,
which is a contradiction. 
\end{proof}

\paragraph{VOs for Static Parts }
\nop{Given a VO $\omega$ and a subset  $\bm N$ of its variables,  
we say that $\calN$ is the {\em neck} of $\omega$ 
if the upper part of $\omega$ is a path consisting 
of the variables in $\bm N$. More precisely, 
there is an ordering $X_1, \ldots, X_n$ of the variables in $\bm N$ such that $X_1$ is 
the root of $\omega$ and for each $i \in [n-1]$, $X_{i + 1}$ is the only child of $X_i$. 
}
For a sequence $\bm N= (X_1, \ldots, X_n)$, we say that a VO $\omega$ has {\em neck} $\bm N$ 
if $X_1$ is the root of $\omega$ and $X_{i + 1}$ is the only child of $X_i$ for each $i \in [n-1]$. 
By Proposition~\ref{prop:safety_properties_q-hierarchical}, the 
dynamic sub-query $\dyn(Q)$ of each well-behaved query is q-hierarchical. 
It is known that each q-hierarchical query admits a well-structured VO~\cite{DBLP:journals/lmcs/KaraNOZ23}. 
Consider a well-behaved conjunctive query $Q$, a well-structured VO $\omega_d$ for $\dyn(Q)$, and a static part $Q_p$ of $Q$ with intersection set $\bm I$. 
We say that a VO $\omega_p$ for $Q_p$ is {\em $\omega_d$-compatible} if there is an ordering 
$\bm N = X_1, \ldots, X_n$ of the variables in $\bm I$ such that (1) 
$i < j$ if and only if $X_i$ appears above $X_j$ in $\omega_d$ for all $i,j \in [n]$, and 
(2) $\omega_p$ has neck $\bm N$.

\begin{example}
\rm 
Figure~\ref{fig:pol_class_construction_appendix}  shows below each static part $Q_p$ 
of the query $Q$ an $\omega_d$-compatible VO for $Q_p$.
For simplicity, the atoms are omitted in the VOs.
In each $\omega_d$-compatible VO, the intersection variables of the corresponding static part form a path and are on top of all other variables.
For instance, the $\omega_d$-compatible VO for the static part 
$Q_p(C) = R(G,C), S(G,F)$ has the intersection variables $\{C, F\}$ on a path on top of the variable  $G$.
\end{example}

\subsection{Proof of Proposition \ref{prop:pol_vo_dynamic_q-hierarchical}}
{\bf Proposition~\ref{prop:pol_vo_dynamic_q-hierarchical}.}
{\it Any well-behaved conjunctive query has a well-structured VO.}

\medskip

Given a well-behaved conjunctive query $Q$, we show how to construct a well-structured VO for $Q$.
The high-level idea of the construction is as follows. 
We start with constructing a well-structured VO $\omega^d$ for $\dyn(Q)$. 
Then, we construct for each static part of $Q$, an $\omega_d$-compatible free-top VO.
Finally, we combine the VOs for the static parts with 
$\omega^d$ and show that the resulting structure  is a well-structured VO for $Q$.

\nop{
Lemma~\ref{lem:vo_neck_construction} states
that  any static part of $Q$ with intersection 
set $\bm I$ admits a free-top VO with neck $\bm I$. Before proving this lemma,
 we illustrate it by an example and formulate an auxiliary lemma.
}

We first illustrate the construction by the following example:
\begin{example}
\rm
Consider again the well-behaved query $Q$ query, its dynamic sub-query $\dyn(Q)$, and its static parts visualized in Figure~\ref{fig:pol_class_construction_appendix}.
The figure shows the well-structured VO $\omega$ for $Q$, the well-structured VO $\omega_d$ for $\dyn(Q)$, and the $\omega_d$-compatible free-top VOs for the static parts. The VO $\omega$ is obtained by combining each VO $\omega_p$ of a static part
$Q_p$ with $\omega_d$ as follows. Let $X$ be the lowest variable in $\omega_p$ that is contained in the intersection set of $Q_p$. The child trees of $X$ in $\omega_p$ become the child trees of $X$ in $\omega_d$. For instance, consider  
the third static part $Q_p(C) = R(G,C), S(G,F)$ with intersection set $\{C,F\}$ in Figure~\ref{fig:pol_class_construction_appendix}. The lowest variable in the VO of $Q_p$
that is contained in the intersection set of $Q_p$ is $F$. This variable has a single child tree consisting of $G$ in $\omega_p$. We add $G$ as a child of $G$ to $\omega_d$. 
\nop{
Observe that in the VO for the third static part $Q_p(C) = R(G,C), S(G,F)$, the free variable $C$ is on top of $F$ and $G$, and the intersection variables $C$ and $F$ are on top of $G$.
}
\end{example}

The next lemma states that a variable in a static part that is not contained in the intersection set $\bm I$ of the static part can be connected to a variable in $\bm I$ via variables that are not contained in $\bm I$:

\begin{lemma}
\label{lem:static_part_reachability}
Let $Q$ be a well-behaved conjunctive query, $Q_p$ a static part of $Q$ with intersection set $\bm I$, $X \in (\vars(Q_p)\setminus \bm I)$, and $Y \in \bm I$. The query $Q$ contains a path $(X=X_1, \ldots, X_n = Y)$ such that the variables $X_1, \ldots, X_{n-1}$ are contained in $\vars(Q)\setminus \bm I$.
\end{lemma}

\begin{proof}
Let $Q$ be a well-behaved conjunctive query, $Q_p$ a static part of $Q$ with intersection set $\bm I$, $X \in (\vars(Q_p)\setminus \bm I)$, and $Y \in \bm I$. By construction, $Q_p$ must contain a static body atom $S(\bm Y)$ such that $\bm Y$ contains $Y$ and a variable $X'$ in $\vars(Q_p)\setminus \bm I$.
The query $Q$ must contain a path $\bm P = (X=X_1, \ldots, X_{n-1} = X')$ such that $X_i \notin \bm I$ for $i \in [n-1]$. Since $Q$ contains the atom $S(\bm Y)$, we can extend
$\bm P$ to a path $\bm P' = (X=X_1, \ldots, X_{n-1} = X', X_n =Y)$ such that 
$X_1, \ldots, X_{n-1}$ are contained in $\vars(Q)\setminus \bm I$.
\end{proof}

The next lemma states that if a static part contains a free variable that is not contained in its intersection set, then all variables in the intersection set must be free:
\begin{lemma}
\label{lem:static_part_free_top}
Let $Q$ be a well-behaved conjunctive query and $Q_p$ a static part of $Q$ with intersection set $\bm I$. If $\free(Q_p) \setminus \bm I \neq \emptyset$,
then $\bm I \subseteq \free(Q)$.
\end{lemma}

\begin{proof}
Consider a well-behaved conjunctive query $Q$ and a static part $Q_p$ of $Q$ 
with intersection set $\bm I$. 
For the sake of contradiction, assume that $\bm I$ contains a bound variable $B$ and $Q_p$ has a free variable $F$ that is not contained in $\bm I$.
By Lemma~\ref{lem:neck_dynamic-coverage}, $Q$ has a dynamic body atom 
$R(\bm X)$ such that $\bm I \subseteq \bm X$.
By Lemma~\ref{lem:static_part_reachability}, $Q$ contains a path 
$\bm P = (F=X_1, \ldots , X_n = B)$ such that none of the variables 
$X_1, \ldots , X_{n-1}$ is contained in $\bm X$. 
This means that $\free(Q) \cap \bm X \cap \bm P = \emptyset$.
This implies that $\bm P$ connects the dynamic body atom $R(\bm X)$ with the head atom but is not head-safe. Hence, 
$Q$ is not well-behaved, which is a contradiction. 
\end{proof}

Using Lemma~\ref{lem:static_part_free_top}, we show that for every static part of a well-behaved query $Q$, we can construct a free-top VO that is compatible with a well-structured VO for the dynamic sub-query of $Q$:

\begin{lemma}
\label{lem:vo_neck_construction}
Let $Q$ be a well-behaved conjunctive query, $\omega_d$ a well-structured VO for $Q$, and $Q_p$ a static part of $Q$. 
The query $Q_p$ has an $\omega_d$-compatible free-top VO.
\end{lemma}

\begin{proof}
Let $Q$ be a well-behaved conjunctive query, $\omega_d$ a well-structured VO for $Q$, and $Q_p$ a static part of $Q$ with intersection set $\bm I$. We show that there is at least 
one $\omega_d$-compatible free-top VO for 
$Q_p$.

First, assume that $\free(Q_p) \subseteq \bm I$. 
Consider the VO $\omega_p$ consisting of a single path whose top-down structure is as follows.
It starts with the variables in $\bm I$ ordered as in $\omega_d$, followed
by the remaining variables of $Q_p$ in arbitrary order. 
The VO $\omega_p$ is an $\omega_d$-compatible VO for $Q_p$. It is free-top, since $\omega_d$ is free-top.
\nop{
First, assume that $\free(Q_p) \subseteq \bm I$. 
Let $\bm O_1$ be an ordering of the variables in $\bm I$  
 that is compatible with $\omega^d$, 
$\bm B_1$ an ordering of the variables in $\bm I \setminus \bm F$ that is compatible with $\omega^d$, and $\bm B_2$ an arbitrary ordering of the variables in $\vars(Q_p)  \setminus \bm I$. Consider
the variable sequence $\bm S = \bm F, \bm B_1, \bm B_2$. 
Arranging the variables in $\bm S$ top-down such that the first variable becomes the root results in a VO for $Q_p$ that is free-top and compatible with $\omega_d$.
}

Now, assume that $\free(Q_p) \not\subseteq \bm I$, which means that
$Q_p$
has a free variable not contained in $\bm I$. 
By Lemma~\ref{lem:static_part_free_top}, it holds $\bm I \subseteq \free(Q_p)$.
We construct a VO $\omega_p$ consisting of a single path whose top-down structure is as follows.
It starts with the variables in $\bm I$ ordered as in $\omega_d$, followed
by the variables in $\free(Q_p) \setminus \bm I$ in arbitrary order, followed by the remaining variables of $Q_p$ in arbitrary order. 
Obviously, $\omega_p$ is a valid VO for $Q_p$ and compatible with $\omega_d$. By construction, it is free-top.
\nop{
Let $\bm F_1$ be an ordering of the variables in  
$\bm I$ that is compatible with $\omega^d$, 
$\bm F_1$ an arbitrary ordering of the variables in $\free(Q_p) \setminus \bm I$, 
and $\bm B$ an arbitrary ordering of the variables in $\vars(Q_p) \setminus \free(Q_p)$. 
Arranging the sequence $\bm S = \bm F_1, \bm F_2, \bm B$. 
top-down results in a free-top VO for $Q_p$ that is compatible with $\omega_d$.
}
\end{proof}

\begin{figure}[t]
	\centering
 \renewcommand{\arraystretch}{1.15}
  \renewcommand{\linenumber}{\makebox[2ex][r]{\rownumber\TAB}}
	\setcounter{magicrownumbers}{0}
  \begin{tabular}{@{\hskip 0.1in}l}
  \toprule
  $\textsf{CreateVO}$(well-behaved CQ $Q$) : well-structured VO for $Q$ \\[0.2ex]
  \midrule
   \linenumber \LET $\omega_d$ be a well-structured VO for $\dyn(Q)$ \\
   \linenumber \LET $Q_1, \ldots , Q_n$ be the static parts of $Q$\\
      \linenumber \LET $\omega_1, \ldots , \omega_n$ be $\omega_d$-compatible free-top VOs for $Q_1, \ldots , Q_n$, respectively\\
   \linenumber \LET $\omega = \omega_d$ \\
    \linenumber \FOREACH $i \in [n]$ \\
    \linenumber \TAB \LET $\omega = \combine(\omega, \omega_i)$ \\
  \linenumber \RETURN $\omega$ \\
  \bottomrule
  \end{tabular}
  \caption{Constructing a well-structured VO for a well-behaved conjunctive query. The existence of a well-structured VO for $\omega_d$ is guaranteed by prior work~\cite{DBLP:journals/lmcs/KaraNOZ23}. The existence of an $\omega_d$-compatible 
  free-top VO for each static part $Q_i$ is guaranteed by Lemma~\ref{lem:vo_neck_construction}. 
  If $\omega$ and $\omega_i$ do not share any variable, then 
  $\combine(\omega, \omega_i)$ is the disjoint union of the two VOs.
  If $\vars(\omega_i) \subseteq \vars(\omega)$, then  $\combine(\omega, \omega_i) = \omega$.
  Otherwise, $\combine(\omega, \omega_i)$ is obtained by making the child trees of variable $X$ in $\omega_i$ to child trees of $X$ in $\omega$, where $X$ is the lowest variable in $\omega_i$
  that is contained in the intersection set of $Q_i$.}
  \label{fig:createVO}
\end{figure}
Given a well-behaved conjunctive query $Q$, the procedure \textsf{CreateVO}
in Figure~\ref{fig:createVO} describes how to construct a well-structured VO for 
$Q$. The procedure first constructs a well-structured VO $\omega_d$  
for $\dyn(Q)$ (Line~1), whose existence is guaranteed by prior work~\cite{DBLP:journals/lmcs/KaraNOZ23}.
Then, it constructs the static parts $Q_1, \ldots , Q_n$ of $Q$ (Line~2).
For each static part $Q_i$, it then constructs an $\omega_d$-compatible free-top VO $\omega_i$ (Line~3), whose existence is guaranteed by Lemma~\ref{lem:vo_neck_construction}. 
The procedure obtains the final VO by combining the VOs $\omega_1, \ldots, \omega_n$ 
with $\omega_d$ (Lines 4--6).
The VO $\combine(\omega, \omega_i)$ in Line~6 is defined as follows. 
  If $\omega$ and $\omega_i$ do not share any variable, then 
  $\combine(\omega, \omega_i)$ is the disjoint union of the two VOs.
  If $\vars(\omega_i) \subseteq \vars(\omega)$, then  $\combine(\omega, \omega_i) = \omega$.
  Otherwise, $\combine(\omega, \omega_i)$ is obtained by making the child trees of variable $X$ in $\omega_i$ to child trees of $X$ in $\omega$, where $X$ is the lowest variable in $\omega_i$
  that is contained in the intersection set of $Q_i$.
\nop{We explain next how a VO $\omega_i$ 
for a static part is attached to $\omega^d$.

By Lemma~\ref{lem:neck_dynamic-coverage}, the neck of $\omega_i$ 
is included in a single atom of $\dyn(Q)$.
Hence, all neck variables are on a root-to-leaf path of 
$\omega^d$. If $\bm Y_i$ is empty, we add $\omega_i$ as a separate tree to $\omega$.
Otherwise, let $X$ be the lowest variable in $\omega_i$ that is contained 
in $\bm Y_i$ and let $Y$ be the lowest variable in $\omega^d$ 
that is contained in $\bm Y_i$.
We make the child trees of $X$ to child trees of $Y$. 
}

\nop{
\begin{example}
\rm
Consider the well-structured VO $\omega$ for the query $Q$, the well-structured VO 
$\omega_d$ for the query $\dyn(Q)$, and the free-top VOs for static parts of $Q$
that are compatible with $\omega_d$ given in Figure~\ref{fig:pol_class_construction_appendix} (second row).
The VO $\omega$ is obtained by joining the Vos for the static parts with $\omega_d$. 
\end{example}
}

The following lemma concludes the proof of Proposition~\ref{prop:pol_vo_dynamic_q-hierarchical}:

\begin{lemma}
\label{lem:createVO_creates_free_top}
For any well-behaved conjunctive query $Q$, $\textsf{CreateVO}(Q)$ is a well-structured VO for $Q$.
\end{lemma}

\begin{proof}
Consider a well-behaved conjunctive query $Q$.
Let  $Q_1, \ldots, Q_n$ be the static parts of $Q$ with intersection sets $\bm I_1, \ldots , \bm I_n$, respectively.
Let $\omega^d$ be the well-structured VO for $\dyn(Q)$ and
$\omega_1, \ldots, \omega_n$ the $\omega^d$-compatible free-top VOs for $Q_1, \ldots, Q_n$ constructed by $\textsf{CreateVO}(Q)$
in Lines 1--2. 
Let $\omega = \textsf{CreateVO}(Q)$. We show that
$\omega$ is a well-structured VO for $Q$.

We first show that $\omega$ is canonical. 
Since $\omega_d$ is canonical, the set of variables of each dynamic body atom $A$ in $Q$ is the set of inner nodes of a root-to-leaf path
in $\omega_d$ that ends with $A$.
Each root-to-leaf path in $\omega_d$ remains a root-to-leaf path in $\omega$. Hence, 
$\omega$ is canonical. 

We now show that the variables of any body atom $R(\bm X)$ of $Q$ is on a root-to-leaf path of $\omega$.     
The atom $R(\bm X)$ is either contained in $dyn(Q)$ or in a static part $Q_i$. 
In the former case, the variables in $\bm X$ lie on a root-to-leaf path in $\omega_d$, 
since $\omega_d$ is a VO for $\dyn(Q)$.
Hence, they are on a root-to leaf path in $\omega$, since 
$\textsf{CreateVO}(Q)$ preserves the structure of $\omega_d$.
In the latter case, the variables in $\bm X$ lie on a root-to-leaf path
of $\omega_i$, since $\omega_i$ is a VO for $Q_i$. 
By construction, all 
variables in the intersection set $\bm I_i$ of $Q_i$ are on a root-to-leaf path in $\omega_i$. 
By Lemma~\ref{lem:neck_dynamic-coverage}, all 
variables in $\bm I_i$ are contained in  a dynamic atom of $Q$. Since $\omega_d$ is a VO for $\dyn(Q)$, the variables in $\bm I_i$ are on a root-to-leaf path
in $\omega_d$. Hence, they are on a root-to-leaf path
in $\omega$. Let $X$ be the lowest variable in $\omega_i$
that is contained in $\bm I_i$. By construction, $X$ is also the lowest variable in $\omega$ that is contained in $\bm I_i$.
The subtrees of $X$ in $\omega_i$ become the subtrees of $X$
in $\omega$. Hence, each set of variables
that are on a root-to-leaf path in $\omega_i$
remain on a root-to-leaf path in $\omega$.

Now, we show that $\omega$ does not contain any bound variable on top of a free variable.  
For the sake of contradiction, assume that $\omega$ contains a bound variable $B$ on top 
of a free variable $F$. Since $\omega_d$ and $\omega_1, \ldots  ,\omega_n$ are free-top and the intersection variables of each $Q_i$ are on top of all other variables of $Q_i$, it must hold 
that $B \in \bm I_i$ 
and $F \in \vars(\omega_i) \setminus \bm I_i$, for some $i \in [n]$.
By Lemma~\ref{lem:static_part_reachability}, $Q$ contains a path $(F=X_1, \ldots, X_n = B)$ such that $X_1, \ldots, X_{n-1}$ are contained in $\vars(Q)\setminus \bm I_i$. This means that 
$Q$ contains a path that connects the head-atom of $Q$ with a 
dynamic body atom but is not head-safe. Hence, $Q$ is not 
well-behaved, which is a contradiction.
\end{proof}


\subsection{Proof of Proposition \ref{prop:acyclic_preprocessing_width}}
{\bf Proposition~\ref{prop:acyclic_preprocessing_width}.}
{\it For any conjunctive query $Q$ in $\lin$, it holds $\fw(Q)= 1$.}
\medskip

Similar to the case of $\pol$-queries, we use 
the procedure $\textsf{CreateVO}(Q)$ in Figure~\ref{fig:createVO} to 
construct a well-structured VO with preprocessing width $1$ for any given 
$\lin$-query $Q$.
Compared to the case of $\pol$-queries, we just need to be more precise 
about the structure of the VOs for the static parts, which are constructed in Line~3 of the procedure
$\textsf{CreateVO}$.
 In the following, we first recall two characterizations of free-connex 
 $\alpha$-acyclic queries that will be useful in the construction of VOs for $\lin$-queries 
 (Section~\ref{sec:characterization_acyclic}).
 Then, we will detail the construction of the VOs for the static parts of $\lin$-queries 
 (Section~\ref{sec:VO_for_static_part}). 
 Finally, we will show that the VO returned by the procedure $\textsf{CreateVO}$ has preprocessing width 1
 (Section~\ref{sec:VO_for_whole_query}).

\subsection{Characterizations of Free-Connex $\alpha$-acyclic Queries}
\label{sec:characterization_acyclic}
In Section~\ref{sec:prelims}, we defined free-connex $\alpha$-acyclic queries 
using join trees. In this section, we give two further characterizations of such queries. 
The first characterization is based on VOs: 
A conjunctive query is 
free-connex $\alpha$-acyclic if and only if it has a (not necessarily canonical)
free-top VO of preprocessing width 1~\cite{DBLP:journals/lmcs/KaraNOZ23}.
A further characterization is based on tree decompositions, which we introduce next.

\begin{definition}  
A tree decomposition $\calT$ of a conjunctive query $Q$ is a pair $(T, \chi)$, where $T$ is a tree with vertices $V(T)$ and $\chi: 
V(T) \rightarrow 2^{\vars(Q)}$ maps each node $v$ of $T$ to a subset $\chi(v)$ of variables of $Q$ such that the following properties hold:
\begin{enumerate}
    \item for every body atom $R(\bm X)$ in $Q$, the variable set $\bm X$ is a subset of $\chi(v)$ for some $v \in V(T)$,
    \item for every variable $X \in \vars(Q)$, the set $\{v | X \in \chi(v)\}$ is a non-empty connected subtree of $T$. 
\end{enumerate}
The sets $\chi(v)$ are called the {\em bags} of the tree decomposition.
A tree decomposition is {\em free-connex} if it has a connected sub-tree
such that the union of the bags of this sub-tree is the set of free variables of $Q$.
\end{definition}  
Given a tree decomposition $\calT = (T, \chi)$ for a conjunctive query $Q$, the {\em fractional hypertree width} 
$\fhtw$ of $\calT$ is defined as $\fhtw(\calT)  = \max_{v \in V(T)} \rho^*_Q(\chi(v))$.
A conjunctive query is free-connex $\alpha$-acyclic if and only if it has a free-connex tree decomposition of fractional hypertree 
width 1. 

It follows from prior work:
\begin{lemma}[\cite{OlteanuZ15}]
\label{lem:tree_dec_to_VO}
 Every tree decomposition for a conjunctive query $Q$ with  fractional hypertree width $1$ 
can be transformed into a VO for $Q$ with preprocessing width $1$.
\end{lemma}

\begin{proof}[Proof sketch]
We briefly sketch the construction. Consider a tree decomposition $\calT$ 
with fractonal hypertree width $1$. 
 For any bag $\bm B$ of $\calT$, we denote by $\bm B'$ the subset of $\bm B$ 
such that a variable $X$ of $Q$ is in $\bm B'$ if and only if $\bm B$ is the highest bag in $\calT$
containing $X$. For each bag $\bm B$ let $\bm P_{\bm B}$ be a top-down path of the variables in 
$\bm B'$ in arbitrary order.  We obtain the desired VO from 
$\calT$ by replacing each bag $\bm B$ by the path $\bm P_{\bm B}$, such that each incoming 
edge to $\bm B$ becomes an incoming edge to the first variable in $\bm P_{\bm B}$ 
and each outgoing edge from $\bm B$ becomes an outgoing edge from the last 
variable in $\bm P_{\bm B}$. 
\end{proof}

Using the translation in the above proof sketch, we can easily show:

\begin{lemma}
\label{free-connex-to-neck}
Consider  a free-connex $\alpha$-acyclic conjunctive query $Q$ containing an atom  $R(\bm X)$ such that  
$\free(Q) \subseteq \bm X$ or $\bm X\subseteq \free(Q)$.
Let $\bm N$ be an ordering of the variables in $\bm X$ such that  the free variables come before the bound ones.
The query $Q$  has a free-top VO with neck $\bm N$ 
and preprocessing width $1$.
\end{lemma}

\begin{proof}
Consider a free-connex $\alpha$-acyclic conjunctive query $Q$ and an atom  $R(\bm X)$ of $Q$.
Let $\bm N$ be an arbitrary ordering of the variables in $\bm X$ such that the free variables come before the bound ones.

First, consider the case that $\free(Q) \subseteq \bm X$.
Consider a tree decomposition $\calT$ for 
$Q$ whose fractional hypertree width is 1 and which is rooted at 
a bag $\bm B$ containing $\bm X$. Let $\omega$ be the VO obtained 
from $\calT$ by the construction described in the proof of Lemma~\ref{lem:tree_dec_to_VO}, 
such that  (1)  the variables from $\bm X$
are on top of the other variables in $\bm B$ and (2) the top-down ordering of the variables from $\bm X$
corresponds to the left-to-right ordering given by $N$.
By construction, $\omega$ has neck $N$. By Lemma~\ref{lem:tree_dec_to_VO},
$\omega$ has preprocessing width 1. Since all free variables are 
contained in $\bm X$, $\omega$ is free-top.
  
Now, consider the case that $\bm X \subseteq \free(Q)$.
Consider  a free-connex tree decomposition $\calT$ 
for $Q$ whose fractional hypertree width is 1.
Let $\calT'$ be the connected subtree of $\calT$ such that the union of the bags of $\calT'$
is the set of free variables of $Q$.
We translate $\calT$ into a (possibly different) free-connex tree decomposition
$\hat{\calT}$ for $Q$ with a designated bag $\hat{\bm B}$.  
First consider the case that 
$\calT'$ has a bag $\bm B$ subsuming $\bm X$.
In this case, we set $\hat{\calT} = \calT$  and $\hat{\bm B} = \bm B$.
Now, consider the case that $\calT'$ does not have a bag containing $\bm X$.
In this case, $\calT$ must have a bag $\bm B$ that contains
the variables in $\bm X$ and is directly connected to a bag $\bm B'$ in $\calT'$.
Let $\hat{\calT}$ be the decomposition that arises from $\calT$ by adding a bag $\hat{\bm B}$ 
between $\bm B'$ and $\bm B$ that consists of the variables in $\bm X$.
The tree $\hat{\calT}$ is a valid free-connex tree-decomposition for $Q$ of fractional 
hypertree width 1. 
We root the tree decomposition $\hat{\calT}$ at the bag $\hat{\bm B}$ and convert it into a VO
$\omega$ as described in the poof of Lemma~\ref{lem:tree_dec_to_VO}, 
such that  (1)  the variables from $\bm X$
are on top of the other variables in $\hat{\bm B}$ and (2) the top-down ordering of the variables from $\bm X$
corresponds to the left-to-right ordering given by $\bm N$.
By construction, $\omega$ has neck $\bm N$ and is free-top. By Lemma~\ref{lem:tree_dec_to_VO},
it has preprocessing width 1.
\end{proof}

\subsection{Constructing VOs for Static Parts}
\label{sec:VO_for_static_part}
Consider a conjunctive query $Q$ in $\lin$
and a well-structured VO $\omega_d$ for $\dyn(Q)$.
Let $Q_p$ be a static part of $Q$ with intersection set 
$\bm I$. We explain how to construct a free-top $\omega_d$-compatible 
VO for $Q_p$.  
By definition, $Q$ is free-connex $\alpha$-acyclic, hence, it has a 
free-top VO of preprocessing width 1.
Let $Q'$ and $Q_p'$ be the queries obtained from $Q$ and respectively $Q_p$ by adding a fresh atom 
$R(\bm I)$. Since $\bm I$ is subsumed by a dynamic atom in $Q$ (Lemma~\ref{lem:neck_dynamic-coverage}),
the query  $Q'$ must be free-connex acyclic either.
Hence, it has a free-top VO of preprocessing width $1$. 
We start with such a VO $\omega'$ for $Q'$
and eliminate one-by-one all variables that do not appear in $Q_p'$. 
If a removed variable $X$ has a  parent variable $Y$, the children of $X$ become the children of $Y$.
If a removed variable $X$ does not have a parent, the child trees of $X$ become independent trees. 
We denote the obtained VO by $\omega_p'$.

\begin{claim}
\label{claim:VO_static_part}
The VO $\omega_p'$ is a free-top VO for $Q_p'$ with preprocessing width $1$.
\end{claim}

\begin{proof}
It follows immediately from construction that the variables of each atom in $Q_p'$ are on a root-to-leaf path 
in $\omega_p'$. Likewise, $\omega_p'$ is free-top. It remains to show that for each variables $X$ in 
$\omega_p'$, the set $\bm S = \{X\} \cup \dep_{\omega_{p}'}(X)$ is subsumed by the variable set of 
a single atom in $Q_p'$. For the sake of contradiction,
assume that we need at least two atoms from $Q_p'$ to cover the variables in  $\bm S$.
It must hold  $\bm S \subseteq \{X\} \cup \dep_{\omega'}(X)$. The query $Q'$ must have an atom 
$T(\bm Y)$ such that $\bm Y$ subsumes the variables in $\{X\} \cup \dep_{\omega'}(X)$. It follows from the construction of 
$Q_p'$ that it contains an atom $T'(\bm Y')$ such that $(\vars(Q_p) \cap \bm Y) \subseteq \bm Y'$.
 Hence, $\bm Y'$ subsumes all variables in $\bm S$, which is a contradiction to our assumption
 that we need at least two atoms from $Q_p'$ to cover the variables in  $\bm S$.
\end{proof}

Claim~\ref{claim:VO_static_part} implies that $Q_p'$ is free-connex $\alpha$-acyclic. 
It follows from Lemma~\ref{lem:static_part_free_top} that 
$\free(Q_p') \subseteq \bm I$ or $\bm I\subseteq \free(Q_p')$.
Let $\bm N = (X_1, \ldots, X_n)$ be an ordering of the variables in $\bm I$ such that 
$i <j$ if and only if $X_i$ is above $X_j$ in $\omega_d$.  
By Lemma~\ref{free-connex-to-neck}, we conclude that 
$Q_p'$  has a free-top $\omega_d$-compatible VO $\omega_p$ with neck $\bm N$ 
and preprocessing width $1$.
The VO $\omega_i$ is a free-top VO for $Q_p$.

\subsection{Constructing a VO for the Whole Query}
\label{sec:VO_for_whole_query}
Consider a conjunctive query $Q$ in $\lin$.
We execute the procedure $\textsf{CreateVO}(Q)$ in 
Figure~\ref{fig:createVO}.
Let $\omega_d$ be a a well-structured VO for $Q$ constructed in Line~1
of the procedure. 
Let $Q_1, \ldots, Q_n$ be the static parts of $Q$ with intersection sets 
$\bm I_1, \ldots , \bm I_n$, respectively.
Let $\omega_1, \ldots, \omega_n$ be free-top $\omega_d$-compatible VOs
for $Q_1, \ldots, Q_n$, respectively, constructed in Line~3 of the procedure 
as described in Section~\ref{sec:VO_for_static_part}. 
 Let $\omega$ be the VO constructed for $Q$ by $\textsf{CreateVO}(Q)$.
 It follows from Lemma~\ref{lem:createVO_creates_free_top} that 
$\omega$ is a well-structured VO for $Q$.
In the following, we show that the  preprocessing width of $\omega$ is
1, which concludes the proof.

Consider a variable $X$ that is contained in a dynamic atom of $Q$.
Since $\omega_d$ is well-structured, all ancestor variables of $X$ are contained  in each dynamic atom below $X$. This implies 
$\rho^*_{Q_X}(\{X\} \cup \dep_{\omega}(X)) =1$.
Consider now a variable $X$ that does not appear in $\dyn(Q)$ but in some static part 
$Q_i$. Let $Q_i'$ be the extension of $Q_i$ by a fresh atom $R(\bm I_i)$.
By Claim~\ref{claim:VO_static_part}, $\omega_i$ is a VO for $Q_i'$ with preprocessing width 1. 
Hence, the set $\bm S = \{X\} \cup \dep_{\omega_i}(X)$ is subsumed by the set of variables of a single atom $A$ in $Q'$.
If $A$ is an atom in $Q_i$, this means that
$\bm S$ is subsumed by the set of variables of the atom $A$ in $Q_i$, hence in $Q$.
If $A$ is the atom $R(\bm I_i)$, this means that
$\bm S$ is subsumed by a dynamic atom of $Q$.
Overall, we derive that the preprocessing width of $\omega$ is $1$.

\subsection{Proof of Proposition \ref{prop:safe_view_trees}}
{\bf Proposition~\ref{prop:safe_view_trees}.}
{\it For any query $Q$ in $\pol$, VO $\omega$ for $Q$, and database of size $N$, $\rewrite(\omega)$ is a safe rewriting for $Q$
    with $\bigO{N^{\fw}}$ computation time, where $\fw$ is the preprocessing width of $Q$. 
}

\medskip
Given a query $Q$ in $\pol$, a VO $\omega$ for $Q$, and database of size $N$, let $\calT$ be the set of view trees 
returned by $\rewrite(\omega)$.
We show that $\calT$ satisfies all properties for safe 
rewritings, as specified in Definition~\ref{def:safe_rewriting}. 

We first show that the {\em correctness properties} for safe rewritings hold. 
By the definition of VOs, the variables of each body atom must be on 
a root-to-leaf path and each atom is placed under its lowest variable in the VO. 
This implies that for each connected component $\calC$ of $Q$, 
all body atoms in $\calC$ must be in a single tree of $\calT$ (first correctness property). Since all atoms containing a variable $X$ must be below $X$ in the VO, all atoms that contain $X$ must be in the subtree rooted at a projection view $V_X$ (second correctness property).

Now we show that the {\em update property} holds for $\calT$.
By definition, $\omega$ is canonical.
Hence, the set of variables of  each dynamic atom contains all ancestor variables. 
This implies that the set of free variables of each dynamic view contains 
all the ancestor variables of that view. Consider a dynamic view $V_X$ and its sibling view $V_Y$. 
We consider two cases. First, assume that  $V_Y$ is an atom.
In this case, the variables of $V_Y$ are included in the set of  ancestor variables of $V_Y$. 
Hence, the set of free variables of $V_X$ contains the set of variables of $V_Y$.
Now, assume that $V_Y$ is a projection view. In this case,  the set of free variables of $V_Y$ can only contain variables that are ancestors of $V_X$. Hence also in this case, the set of free variables of $V_X$ covers all variables of $V_Y$.

The {\em enumeration property} follows simply from the fact that the $\calT$ follows a free-top VO. 

It remains to show that all views in $\calT$ can be computed in $\bigO{N^{\fw}}$ time.
The proof follows closely prior work using view trees~\cite{ICDT23_access_pattern}. Let $T$ be a tree in $\calT$. 
We show by induction on the structure of $T$   that every view in $T$
can be computed in $\bigO{N^{\fw}}$ time.

\smallskip
\textit{Base Case}:
Each leaf atom can obviously be computed in 
$\bigO{N^{\fw}}$ time. 

\smallskip
\textit{Induction Step}:
Consider a projection view a $V_X'(\dep(X))$
in $T$.   
Such a view results from its single child view by projecting away 
$X$.
By induction hypothesis,  the view $V_X$ can be computed in $\bigO{N^{\fw}}$ time. Hence, it is of size $\bigO{N^{\fw}}$. 
The view $V_X'$ can be constructed from $V_X$ by a single scan, which takes  
$\bigO{N^{\fw}}$ time.

Consider now a join view  $V_X(\{X\} \cup \dep(X))$ in $T$.   
Let $V_1(\calS_1), \ldots,  V_k(\calS_k)$ be the child 
views of $V_X$.
By induction hypothesis, each child view can be computed 
in $\bigO{N^{\fw}}$ time.
By construction of $T$, any variable that appears in at least 
two of the child views must be contained in the schema of $V_X$.
This mean that variables that do not appear in $\{X\} \cup \dep(X)$
cannot be join variables among the child views of $V_X$.
In each child view we project away 
all non-join variables that do not appear in $\{X\} \cup \dep(X)$, 
using $\bigO{N^{\fw}}$ time.  
Let $V_1'(\calS_1'), \ldots ,V_k'(\calS_k')$ be the resulting 
child views. 
Using a worst-case optimal join algorithm~\cite{Ngo:SIGREC:2013}, we 
then compute the view $V_X$ from its child views 
in $\bigO{|V_X|}$ time.
The size of $V_X$ is upper-bounded by $\bigO{N^{p}}$ where 
$p = \rho^*_{Q_X}(\{X\} \cup \dep(X))$
and $Q_X$ is the query that joins 
all atoms in the subtree rooted at $V_X$.
Since $p \leq \fw$, the view $V_X$ can be computed in $\bigO{N^{\fw}}$ time.

%% file: app_evaluation_expo.tex
\section{Missing Details in Section~\ref{sec:evaluation_expo}}
\label{app:evaluation_expo}
\subsection{Proof of Proposition \ref{prop:max_dynamic_database}}
{\bf Proposition~\ref{prop:max_dynamic_database}.}
{\it     For any conjunctive query $Q$ in $\expo \setminus \pol$ and database $D$ of size $N$,
    the maximum dynamic database $D^d_{\max}$ for query $Q$ and database $D$ has   $\bigO{N^{\rho^*(\stat(Q))}}$ size. 
}

\medskip 

Consider a conjunctive query $Q$ in $\expo \setminus \pol$ and database $D$ of size $N$.
Given a dynamic body atom $R^d(\bm X)$, let $Q_R'$ be the variant of $Q_R$, where all variables are free. 
The size of the result of $Q_R'$ is bounded by $\rho^*(\stat(Q')$ \cite{AtseriasGM13}.
We obtain the result of $Q_R$ by projecting the result of $Q_R'$ onto $\bm X$.
Hence, the size of the result of  $Q_R$ is bounded by $\rho^*(\stat(Q')$.
Since $D^d_{\max}$ consists of the results of the queries $Q_R$ for constantly many 
dynamic relations $R^d$, the size of $D^d_{\max}$ is $\bigO{N^{\rho^*(\stat(Q))}}$.

%% file: app_lower_bounds.tex
\section{Missing Details in Section~\ref{sec:lower_bounds}}
\label{app:lower_bounds}
We prove the lower bound in Theorem~\ref{thm:dichotomy_linear}, which states:

\begin{proposition}
\label{prop:lower_bound_dichotomy}
Let a conjunctive query $Q$ and a database of size $N$.
If $Q$ is not in $\lin$ and does not have repeating relation symbols, then it  
cannot be maintained
with $\bigO{N}$ preprocessing time, $\bigO{1}$
update time, and $\bigO{1}$ enumeration delay,
unless the OMv or the BMM conjecture fail. 
\end{proposition}

In Sections~\ref{sec:lower_bound_RST} and \ref{sec:lower_bound_ST}, we give lower bounds on the complexity of evaluating the 
queries $Q_{RST}() =R^d(A),S^s(A,B),T^d(B)$ and respectively $Q_{ST}(A) = S^s(A,B),T^d(B)$. 
In Section~\ref{sec:general_lower_bound}, we show a lower
bound on the complexity of evaluating arbitrary conjunctive queries that are not contained in $\lin$ and do
not have repeating relation symbols.

\subsection{Lower Bound for $Q_{RST}$}
\label{sec:lower_bound_RST}
\nop{
The lower bound for the query $Q_{RST}()$ $=$ $R^d(A),$ $S^s(A,B),$ $T^d(B)$ relies on the OuMv conjecture, which is implied by the OMv conjecture~\cite{Henzinger:OMv:2015}:
}
{\bf Proposition~\ref{prop:Q_RST_hard}.}
{\it 
There is no algorithm that evaluates the conjunctive query 
$Q_{RST}()$ $=$ $R^d(A),$ $S^s(A,B),$ $T^d(B)$ 
with 
$\bigO{N^{3/2-\gamma}}$ preprocessing time, 
$\bigO{N^{1/2-\gamma}}$ update time, and 
$\bigO{N^{1/2-\gamma}}$ enumeration delay
for any $\gamma >0$, where $N$ is the database size, unless the OuMv conjecture fails.
}

\medskip
We detail the proof sketch of Proposition~\ref{prop:Q_RST_hard} in Section~\ref{sec:lower_bounds}.
Assume that we have an Algorithm $\calA$ that maintains the  
query $Q_{RST}$ with $\bigO{N^{3/2 - \gamma}}$ preprocessing time, 
$\bigO{N^{1/2 - \gamma}}$ update time, and $\bigO{N^{1/2 - \gamma}}$ enumeration delay
for some $\gamma >0$. 
We show that we can use Algorithm $\calA$ to solve the OuMv problem given in Definition \ref{def:OuMV_problem} in subcubic time, which contradicts Conjecture \ref{conj:OuMV}. 

Consider an $n$-by-$n$ matrix $M$ and $n$ pairs of $n$-dimensional vectors $(u_1, v_1), \ldots , (u_n, v_n)$ that serve as the input to the OuMv problem. We construct an Algorithm $\calB$ that uses the static relation $S^s(A, B)$ to encode $M$ and the two dynamic relations $R^d(A)$ and $T^d(B)$ to encode the vectors 
$u_r$ and respectively $v_r$ for $r\in[n]$. Next, we explain the encoding in detail. 

Algorithm $\calB$ starts with the empty relations $R^d$, $S^s$, and $T^d$. First, it populates relation $S^s$ such that $(i,j) \in S^s$ if and only if $M(i,j) = 1$. 
Then, it executes the preprocessing mechanism of Algorithm $\calA$. 
In each round $r \in [n]$, Algorithm $\calB$ receives the vector pair $(u_r, v_r)$ and updates $R^d(A)$ and $T^d(B)$ so that $i \in R^d$ if and only if $u_r[i]= 1$ and $i \in T^d$ if and only if $v_r[i]= 1$. We observe that $u^T_rMv_r = 1$ if and only if the (Boolean) result of $Q_{RST}$ is true. Algorithm $\calB$ triggers the enumeration mechanism of $\calA$ and outputs $1$ if the result of $Q_{RST}$ is true. Otherwise, it outputs $0$. 

We analyse the overall time used by Algorithm $\calB$. 
Given that $M$ is an $n\times n$ matrix  $M$, the size and construction time of relation $S^s$ are both $\bigO{n^2}$. It results in a database of size $\bigO{n^2}$.
The preprocessing time is $\bigO{(n^2)^{3/2 - \gamma}} = \bigO{n^{3 - 2\gamma}}$.
In each round $r \in [n]$, Algorithm $\calB$ executes at most 
$4n$ updates to $R^d$ and $T^d$ and checks emptiness for $Q_{RST}$. 
Since the database size remains $\bigO{n^2}$, the time to update the relations is
$\bigO{4n \cdot (n^2)^{0.5 - \gamma}} = \bigO{n^{2 - 2\gamma}}$.
It is sufficient to enumerate the first result tuple to check if $Q_{RST}$ is empty. 
The time to do the emptiness check is
$\bigO{(n^2)^{0.5 - \gamma}} = \bigO{n^{1 - 2\gamma}}$.
Hence, for $n$ rounds, the overall time is $\bigO{n^{3-2\gamma}}$. 
We conclude that Algorithm $\calB$ takes $\bigO{n^{3-2\gamma}}$ time to 
solve the OuMv problem, which contradicts Conjecture \ref{conj:OuMV}.

\subsection{Lower Bound for the Query $Q_{ST}$}
\label{sec:lower_bound_ST}
\begin{proposition}
\label{prop:Q_ST_hard}
There is no algorithm that evaluates the conjunctive query 
$Q_{ST}(A) = S^s(A, B), T^d(B)$ 
with 
$\bigO{N^{3/2-\gamma}}$ preprocessing time, 
$\bigO{N^{1/2-\gamma}}$ update time, and 
$\bigO{N^{1/2-\gamma}}$ enumeration delay
for any $\gamma >0$, where $N$ is the database size, unless the OMv conjecture fails.
\end{proposition}

\medskip
The proof is similar to that of Proposition \ref{prop:Q_RST_hard}. 
%
Assume there is an Algorithm $\calA$ that maintains the  
query $Q_{ST}$ with 
$\bigO{N^{3/2-\gamma}}$ preprocessing time, 
$\bigO{N^{1/2-\gamma}}$ update time, and 
$\bigO{N^{1/2-\gamma}}$ enumeration delay
for some $\gamma >0$.
We show that the existence of Algorithm $\calA$ contradicts Conjecture \ref{conj:OMV}. 

Let $M$ be an $n$-by-$n$ matrix  and $v_1, \ldots , v_n$ $n$-dimensional vectors that serve as input to the 
OMv problem. 
We design an algorithm $\calB$ that first encodes $M$ into $R^d$.
Then, it executes the preprocessing mechanism of Algorithm $\calA$.
In each round $r \in [n]$, it executes at most $2n$ updates to encode the incoming vector $v_r$ into $T^d$. 
To obtain the result of $Mv_r$, Algorithm $\calB$ enumerates the result of $Q_{ST}$.
The preprocessing time of Algorithm $\calB$ is
$\bigO{(n^2)^{1.5 - \gamma}} = \bigO{n^{3 - 2\gamma}}$.
The total update and enumeration time over $n$ rounds is,
$\bigO{n \cdot n \cdot (n^2)^{0.5 - \gamma}} = \bigO{n^{3 - 2\gamma}}$.
The subcubic processing time contradicts Conjecture \ref{conj:OMV}.

\subsection{Proof of Proposition \ref{prop:lower_bound_dichotomy}}
\label{sec:general_lower_bound}
Consider a conjunctive query $Q \notin \lin$ without repeating relation symbols. 
By definition of the class $\lin$, one of the following three cases holds: (1) $Q$ is not 
free-connex $\alpha$-acyclic, or (2) 
it has a path connecting two dynamic body atoms that is not body-safe, or (3) it has a
path connecting a dynamic body atom with the head atom that is not head-safe.
If Case (1) holds, then $Q$ does not admit constant-delay enumeration after liner time preprocessing even without executing any update,  unless the Boolean Matrix Multiplication conjecture fails~\cite{BaganDG07}.  
In Case (2), we reduce the evaluation
of $Q_{RST}$ to the evaluation of $Q$. In Case (3), we reduce the evaluation of $Q_{ST}$ to the evaluation of $Q$. 
These reductions transfer the lower bounds for $Q_{RST}$ and respectively $Q_{ST}$ to $Q$. 
In the following two paragraphs, we explain these reductions in more detail. 

\paragraph*{Lower Bound for Queries That Are Not Body-Safe}
Consider a conjunctive query $Q$ that
has a path connecting two dynamic body atoms that is not body-safe.
This means that the query has a path 
$\bm P = (X_1, \ldots , X_n)$ that connects two dynamic 
body atoms $\hat{R}^d(\bm X)$ and $\hat{T}^d(\bm Y)$
such that $\bm P \cap \bm X \cap \bm Y = \emptyset$.
Let $\bm P$ be the shortest path with this property. 
Hence, it must hold that (1) every variable appears at most once in $\bm P$ (by the definition of paths), (2)
 $n \geq 2$, and (3) $X_i \notin \bm X$ and $X_i \notin \bm Y$ for all $X_i$ with $i \in \{2, \ldots, n-1\}$.
\nop{%

\begin{figure}[htbp]
    \centering
    \begin{tikzpicture}
    \node[draw=none,fill=none](X1) at (0, 0) {$X_1$};
    \node[draw=none,fill=none](dotsX) at (-1.5, 0) {$\dots$};
    \node[draw=none,fill=none,rotate=90](dotsXv) at (0, -.6) {$\dots$};
    \node[draw=none,fill=none](Xn) at (8, 0) {$X_n$};
    \node[draw=none,fill=none](dotsY) at (9.5, 0) {$\dots$};
    \node[draw=none,fill=none,rotate=90](dotsYv) at (8, -.6)  {$\dots$};
    \node[draw=none,fill=none](X2) at (2, 0) {$X_2$};
    \node[draw=none,fill=none](dots1) at (4, 0) {$\dots$};
    \node[draw=none,fill=none](Xn-1) at (6, 0) {$X_{n-1}$};
    \node[draw=none,fill=none](r_location) at (-1, .7) {\textcolor{red}{$\hat{R}^d(\bm A)$}};
    \draw[red] (-1,0) circle [x radius= 1.4, y radius=0.4];
    \draw (0,-.3) circle [x radius= .7, y radius=0.4,rotate=90];
    \draw (8,-.3) circle [x radius= .7, y radius=0.4,rotate=90];
    \draw (1,0) circle [x radius= 1.4, y radius=0.4];
    \draw (3,0) circle [x radius= 1.4, y radius=0.4];
    \draw (5,0) circle [x radius= 1.4, y radius=0.4];
    \draw (6.9,0) circle [x radius= 1.5, y radius=0.4];
    \draw[red] (9,0) circle [x radius= 1.4, y radius=0.4];
    \node[draw=none,fill=none](r_location) at (9, .7) {\textcolor{red}{$\hat{T}^d(\bm B)$}};
    \end{tikzpicture}
    \caption{Illustration of a path in a query that does not have safe atom-to-atom paths.}
    \label{fig:atom_violation}
\end{figure}
}
The reduction of the evaluation of the query 
$Q_{RST}() = R^d(A),S^s(A,B),T^d(B)$ to the evaluation of the query $Q$ works as follows. 
We use the 
values of the variables 
$X_1$ and $X_2$ in all body atoms of $Q$, besides the atoms $\hat{R}^d$ and $\hat{T}^d$, 
  to encode the values of the variables $A$ and respectively $B$ in $Q_{RST}$.
   The values of the variables 
$X_3, \ldots , X_n$ become copies of the values of $X_2$.
We use the values of the variable $X_1$ in $\hat{R}^d$ to encode the values 
of the variable $A$ in $R^d$ and use the values of the variable $X_n$ in $\hat{T}^d$ to encode the values 
of the variable $B$ in $T^d$.    
All other variables in $Q$ are assigned a fixed dummy value. 

In the preprocessing stage, we construct all relations besides 
$\hat{R}^d$ and $\hat{T}^d$. 
Each update to $R^d$ or $T^d$ is translated 
into an update to $\hat{R}^d$ or respectively $\hat{T}^d$. 
Each enumeration request to $Q_{RST}$ is translated into an enumeration request to $Q$. The answer of  $Q_{RST}$ is true
if and only if the result of $Q$ is non-empty.

\paragraph*{Lower Bound for Queries That Are Not Head-Safe}
Consider now a conjunctive query $Q$ that
has a path connecting the head atom with a body atom that is not head-safe.
This means that the query has a path 
$\bm P = (X_1, \ldots , X_n)$ such that $X_1$ is free and $X_n$ 
is contained in a dynamic body atom $\hat{T}^d(\bm Y)$ 
such that $\bm P \cap \bm \free(Q) \cap \bm Y = \emptyset$.
Let $\bm P$ be the shortest path with this property. 
It must hold that (1) every variable appears at most once in $\bm P$, (2)
 $n \geq 2$, (3) all variables $X_2, \ldots , X_n$ are bound, and (4)
 $X_i \notin \bm Y$ for $i \in \{1, \ldots, n-1\}$.
The reduction of the evaluation of  
$Q_{ST}(A) = S^s(A,B),T^d(B)$ to the evaluation of the query $Q$ is similar to the reduction in case of the 
query
$Q_{RST}$. 
We use the 
values of the variables 
$X_1$ and $X_2$ in all body atom of $Q$, besides the atom $\hat{T}^d$, 
  to encode the values of the variables $A$ and respectively $B$ of the atom $S^s$ in $Q_{ST}$.
   The values of the variables 
$X_3, \ldots , X_n$ become copies of the values of $X_2$.
We use the values of the variable $X_n$ in $\hat{T}^d$ to encode the values 
of the variable $B$ in $T^d$. We assign a fixed dummy value  to all   
 other variables in $Q$.
In the preprocessing stage, we construct all relations besides $\hat{T}^d$. 
Each update to $T^d$ is translated 
into an update to the relation $\hat{T}^d$. 
Each enumeration request to the query $Q_{ST}$ is translated into an enumeration request to $Q$. The answer of  $Q_{ST}$ is equal to the result of $Q$ projected onto $X_1$. Since all free variables of $Q$ besides $X_1$ are assigned a fixed dummy, the 
projection of the result tuples onto $X_1$ does not increase the enumeration delay for the $X_1$-values.